\documentclass[a4paper,USenglish,thm-restate]{lipics-v2019}
\usepackage{comment}
\usepackage[normalem]{ulem}
\usepackage{algorithm,amsmath,amsthm,amsfonts,amssymb}
\usepackage{thmtools,thm-restate}
\usepackage{paralist}
\usepackage{multicol}
\usepackage{hyperref}
\usepackage{comment}

\usepackage{lineno}

\usepackage[noend]{algorithmic}

\usepackage{graphicx}

\newtheorem{invariant}[theorem]{Invariant}

\newcommand{\EQ}{\texttt{EQ}}
\newcommand{\alg}{\textsc{AC-ASO}}
\newcommand{\commentOut}[1]{}
\newcommand{\tuple}[1]{\langle #1 \rangle}

\newcommand{\xiong}[1]{{\color{purple} #1}}
\newcommand{\lewis}[1]{{\color{blue} #1}}

\newcommand{\remove}[1]{}
\newcommand{\quotes}[1]{``#1"}
\newcommand{\scan}[0]{{\sc Scan~}}
\newcommand{\update}[0]{{\sc Update}}

\nolinenumbers

\title{Amortized Constant Round Atomic Snapshot in Message Passing Systems} %TODO Please add

\titlerunning{Amortized Constant Round Atomic Snapshot in Message-Passing Systems} %TODO optional, please use if title is longer than one line

\author{Vijay Garg}{University of Texas at Austin, USA
}{garg@ece.utexas.edu}{}{}{}

\author{Saptaparni Kumar}{Boston College, USA 
}{kumargh@bc.edu}{}{}{}

\author{Lewis Tseng}{Boston College, USA 
}{lewis.tseng@bc.edu}{}{}{}

\author{Xiong Zheng}{University of Texas at Austin, USA
}{zhengxiongtym@utexas.edu}{}{}{}

\authorrunning{V. K. Garg, S. Kumar, L. Tseng and X. Zheng} %TODO mandatory. First: Use abbreviated first/middle names. Second (only in severe cases): Use first author plus 'et al.'

\Copyright{Vijay Garg, Saptaparni Kumar, Lewis Tseng and Xiong Zheng} %TODO mandatory, please use full first names. LIPIcs license is "CC-BY";  http://creativecommons.org/licenses/by/3.0/


 %TODO mandatory; please add comma-separated list of keywords

\category{} %optional, e.g. invited paper

\relatedversion{} %optional, e.g. full version hosted on arXiv, HAL, or other respository/website
\supplement{}%optional, e.g. related research data, source code, ... hosted on a repository like zenodo, figshare, GitHub, ...

\begin{document}

\maketitle

%TODO mandatory: add short abstract of the document
\begin{abstract}
 We study the lattice agreement (LA) and atomic snapshot problems in asynchronous message-passing systems where up to $f$ nodes may crash. Our main result is a crash-tolerant atomic snapshot algorithm with \textit{amortized constant round complexity}. To the best of our knowledge, the best prior result is given by Delporte et al. [TPDS, 18] with amortized $O(n)$ complexity if there are more scans than updates. Our algorithm achieves amortized constant round if there are $\Omega(\sqrt{k})$ operations, where $k$ is the number of actual failures in an execution and is bounded by $f$. Moreover, when there is no failure, our algorithm has $O(1)$ round complexity unconditionally.
 
 To achieve amortized constant round complexity, we devise a simple \textit{early-stopping} lattice agreement algorithm and use it to ``order'' the update and scan operations for our snapshot object. Our LA algorithm has $O(\sqrt{k})$ round complexity. It is the first early-stopping LA algorithm in asynchronous systems. 

\end{abstract}

% \subsection{Motivation}
\section{Introduction}
\label{s:intro}
The lattice agreement (LA) problem \cite{attiya1995atomic} and atomic snapshot object (ASO) problem~\cite{AfekADGMS93} are two closely related problems in the literature. In the LA problem, given input values from a lattice, nodes have to decide output values that lie on a chain of the input lattice and satisfy some non-trivial validity property. The atomic snapshot object is a concurrent object well studied in shared memory, e.g., ~\cite{AfekADGMS93, Anderson93, aspnes1990wait}. An atomic snapshot object is partitioned into segments. Node $i$ can either update the $i$-th segment (single-writer model), or instantaneously scan \textit{all segments} of the object. In shared memory, Attiya et al.~\cite{attiya1995atomic} showed how to apply algorithms for one problem to solve the another problem. 

% The atomic snapshot problem in a distributed system requires a protocol that returns 
% linearizable responses to the {\em update} and {\em scan} operations on distributed replicated objects. An {\em update} operation changes the value of an object from the old value to the new value.
% A scan operation returns the value of a set of objects.

%The lattice agreement (LA) problem has applications in implementing a special class of replicated state machines: update-query state machines \cite{faleiro2012generalized}, which only supports two types of operations: update and query. It does not support update query mixed operation. It assumes that all updates are commutative. This assumption is reasonable in applications where clients can only modify its own private data but can read the data of other clients. In addition, many data structures can be designed as conflict-free replicated data types (CRDT) \cite{shapiro2011conflict, shapiro2011convergent} which guarantees commutative updates. Skrzypczak et al.~\cite{skrzypczak2019linearizable} shows how to provide linearizability for CRDT using generalized lattice agreement (multi-shot version of lattice agreement). 

% The well known ABD algorithm given by Attiya et al. in \cite{AttiyaBD1995} gives an implementation of atomic register on top of crash-prone message passing systems. 

Both LA algorithms and atomic snapshot objects have a wide spectrum of applications. For example, LA algorithms can be used to implement an update-query state machines \cite{faleiro2012generalized} and linearizable conflict-free replicated data types (CRDT) \cite{skrzypczak2019linearizable}.
Atomic snapshot objects can be used for solving approximate agreement~\cite{attiya1994wait}, randomized consensus~\cite{Aspnes93,AspnesH90}, and implementing wait-free data structures in shared memory~\cite{Herlihy91,AspnesH90}. 
%Atomic snapshot object has applications in both shared memory systems and message passing systems. In shared memory systems, it has applications in solving approximate agreement~\cite{attiya1994wait}, randomized consensus~\cite{Aspnes93,AspnesH90}, and implementing wait free data structures in shared memory systems~\cite{Herlihy91,AspnesH90}. 
%An atomic snapshot object implementation on top a message-passing systems provides a consistent view of the contents of a simulated multi-location shared-memory without preventing concurrent updates to these locations. Such snapshot objects 
In message-passing systems, atomic snapshot objects can be used for
%for %checkpointing and rollback recovery in case of memory corruption to create 
creating self stabilizing memory, and detecting  stable properties to debug distributed programs. In essence, ASO simplify the design and verification of many distributed and concurrent algorithms (example applications  can be found in \cite{taubenfeld2006synchronization},\cite{raynal2012concurrent} and \cite{AttiyaW04}). 
Recently, Guerraoui et al.~\cite{guerraoui2019consensus} also demonstrated a mechanism to use ASO for cryptocurrency.

\vspace{3pt}
\noindent{\bf Contribution}:~~
Our main contribution and closely-related works in message-passing networks are summarized in the table below. %\ref{tab:related_work}. 
%We discuss more related works in Section \ref{section:related}. 
All of our algorithms are proven correct when up to $f < \frac{n}{2}$ nodes may crash. Our LA algorithm is \textit{early-stopping} in the sense that the round complexity depends only on $k~(k \leq f)$, the actual number of failures in an execution.
% Intuitively, ``round'' quantifies the largest message delay, which is formally defined in Section \ref{s:ELA}. 
We also present a general transformation to implement ASO from any LA algorithm. Combined with the $O(\log n)$ LA algorithm in~\cite{zheng2018linearizable}, we obtain an ASO implementation that takes $O(\log n)$ rounds for both \update{} and {\sc Scan}. Our primary contribution is an ASO algorithm which has amortized constant round complexity when there are $\Omega(\sqrt{k})$ operations for each node and incurs only constant message size overhead. As a byproduct, we obtain a linearizable update-query state machines~\cite{faleiro2012generalized} implementation that takes amortized $O(1)$ rounds for each update and query operation and $O(1)$ message size overhead. 
%\lewis{whereas prior algorithms have $O(n)$. XXX Need to check the second part. If we are not sure, we can remove it. XXX} \xiong{Lewis: the constant message size overhead is compared with the transformation algorithm. I am not sure whether we should say constant message size overhead or just measure the communication bits. It seems like communication bits are more standard. Maybe we just don't discuss it here and discuss it when we present AC-ASO}

\begin{table*}[h]
%\caption{Previous Work and Our Results \label{tab:related_work}}
\centering
\begin{tabular}{ |c|c|c|} 
 \hline
 \textbf{Problem}  &\textbf{Reference} &  \textbf{Round Complexity}  \\ \hline
 \multirow{2}{*}{LA} & \cite{zheng2018linearizable} &  $O(\log f)$ \\ 
 \cline{2-3}
 & this paper & $O(\sqrt{k})$\\
 \cline{1-3}
 
%  \multirow{3}{*}{BLA} & \cite{di2019byzantine} & $O(f)$ & $f < \frac{n}{3}$ \\
%  \cline{2-4}
%  & \cite{zheng2020byzantineAsync} & $O(\log f)$ & $f < \frac{n}{5}$ \\   
%  \cline{2-4}
%  & This paper & $O(k)$ & $f < \frac{n}{3}$\\
%  \hline
\multirow{3}{*}{ASO} & \cite{delporte2018implementing}  & \begin{tabular}[c]{@{}l@{}}$O(1)$ for \update{}, $O(n)$ for \scan\end{tabular}\\
 \cline{2-3} 
 & LA \cite{zheng2018linearizable} + transformation [this paper] & $O(\log n)$ for both \update{} and \scan \\
 \cline{2-3}
 & this paper & \begin{tabular}[c]{@{}l@{}}amortized $O(1)$ for both \update{} and \scan\end{tabular}\\
\hline 
\end{tabular}
\vspace*{-0.2in}
\end{table*}

% $O(n^2)$ for update, $O(n^3)$ for scan 

% \lewis{Do we need to cite store-collect paper? But the round complexity of that one is not obvious in our framework? Other round complexity seems to be a more precise name compared to message delays.} \xiong{We don't have to. Their model is dynamic and the store-collect object is a much weaker object than atomic snapshot object}. \lewis{They provide an atomic snapshot implementation based on store-collect (in Algorithm 7). Seems relevant to me.} \xiong{Then, maybe we can discuss it in related work. If I understand correctly, their atomic snapshot algorithm is essentially the double collect algorithm. } \lewis{Agree on the double-collection comment. Yes, we should include it related work. I think there is a high chance for Attiya to review our paper. :) -- She's the chair, and this is some topic that she likes.} 
%\section{Model and Preliminaries}
%\section{Model and Preliminaries}\label{section:model}
% In this section we describe the system assumptions and give a formal definitions of the lattice agreement problem and the atomic snapshot object.

\noindent\textbf{Related Work: Lattice Agreement}~~
The lattice agreement (LA) problem is well studied both in synchronous (e.g., \cite{attiya1995atomic,mavronicolasabound,zheng2018lattice}) and asynchronous (e.g., \cite{zheng2018lattice,zheng2018lattice,faleiro2012generalized}) message-passing systems with crash failures. 
Mavronicolasa et al.~\cite{mavronicolasabound} give an early-stopping algorithm with round complexity of $O(\min\{h, \sqrt{f}\})$, where $h$ is the height of the input lattice. This is the \textit{only} early-stopping LA algorithm that we know before our work. 
In asynchronous systems, the lattice agreement problem cannot be solved when $f \geq \frac{n}{2}$. All existing work assume that $f < \frac{n}{2}$. Faleiro et al.~\cite{faleiro2012generalized} give the first algorithm for this problem which takes $O(n)$ rounds. Xiong et al~\cite{zheng2018linearizable} present an algorithm with round complexity of $O(\log f)$.
LA in the Byzantine fault model is also studied recently. Algorithms for both synchronous systems and asynchronous have been proposed~\cite{zheng2019byzantineSync, zheng2020byzantineAsync, di2020synchronous, di2019byzantine}.  The \textit{equivalence quorum} technique for our LA algorithm is quite different from the techniques in these papers.

\vspace{3pt}
\noindent{\bf Related Work: Atomic Snapshot}~~
ASO is well studied in shared memory, e.g., \cite{AfekADGMS93,attiya1995atomic,attiya1998atomic,inoue1994linear}. Due to space constraint, we focus our discussion in message-passing networks.
In message-passing systems, there are many algorithms for implementing atomic read/write registers in the presence of crash faults~\cite{AttiyaBD1995,Lynch96, AttiyaW04,Attiya2000}. 
A simple way to implement an atomic snapshot object is to first build $n$ SWMR (single-writer/multi-reader) atomic registers, and then use a shared-memory ASO algorithm, e.g., \cite{AfekADGMS93,attiya1995atomic,attiya1998atomic}. 

Delporte et al.~\cite{delporte2018implementing} present the first algorithm for directly implementing an ASO in crash-prone asynchronous message-passing systems. In their implementation, each \update{} operation takes two rounds and $O(n)$ messages, and each \scan operation takes $O(n)$ rounds and $O(n^2)$ messages. 
A recent preprint by Attiya et al.~\cite{attiya2020store} implement a store-collect object in dynamic networks with continuous churn. They also show how to use the store-collect object to build an ASO. These two algorithms and the read/write-register-based algorithms have worse round complexity than our ASO algorithm. In terms of techniques, our ASO algorithm is inspired by \cite{attiya1995atomic}. We will discuss in more details in Section \ref{s:aso}.

%\subsection{System Model}
\vspace{3pt}
\noindent{\bf System Model}:~~
We consider an asynchronous message-passing system composed of $n$ nodes with unique identifiers from $\{1,2,...,n\}$. 
% \vijay{Is $\Pi$ used later?}
Nodes do not have clocks and cannot determine the current time nor directly measure how much time has elapsed since some event. Each node has exactly one server thread and at most one client thread. Client threads invoke \scan or \update{} operations. Each client thread can have at most one \scan or \update{} operation at any time, i.e., each process is sequential. 
Server threads handle incoming messages (i.e., event-driven message handlers).
Local computation is negligible compared to the message delay (or network latency). At most $f$ nodes may fail by crashing in the system. We use $k$, where $k\le f$, to denote the \textit{actual number of failures} in a given execution. 

% We use $C$ to denote the set of correct nodes in the system and

% When presenting out algorithms, we use subscript to indicate the node. 

Each pair of nodes can communicate with each other by sending messages along point-to-point channels. Channels are reliable and FIFO (First-In, First-Out). ``Reliable'' means that a message $m$ sent by node $i$ to node $j$ is eventually received by node $j$ if $j$ has not already crashed. That is, once the command ``send $m$ to $j$'' is completed at node $i$, then the network layer is responsible for delivering $m$ to $j$. The delivery will occur even if node $i$ crashes after completing the ``send'' command. Such a channel can be implemented by a reliable broadcast primitive in practical networks \cite{bracha1987asynchronous}. 
FIFO means that if message $m_1$ is sent before message $m_2$ by node $i$ to node $j$, then $m_1$ is delivered  before $m_2$ at node $j$.

\vspace{3pt}
\noindent{\bf Lattice Agreement (LA)}:~~
% There are two versions of the lattice agreement problem:  one-shot lattice agreement and long-lived lattice agreement. Please refer to Section~\ref{section:related} for related work on lattice agreement. 
Let ($X$, $\leq$, $\sqcup$) be a finite join semi-lattice with a partial order $\leq$ and join $\sqcup$. Two values $u$ and $v$ in $X$ are comparable iff $u \leq v$ or $v \leq u$. The join of $u$ and $v$ is denoted as $\sqcup \{u, v\}$. $X$ is a \textit{join semi-lattice} if a join exists for every non-empty finite subset of $X$. In this paper, we use the term {\em lattice} instead of
{\em join semi-lattice} for simplicity. More background on join semi-lattices can be found in~\cite{garg2015Introduction}.

In the lattice agreement problem~\cite{attiya1995atomic}, each node $i$ proposes a value $x_i \in X$ and must decide on some output $y_i \in X$ such that the following properties are satisfied: 
% \vspace{-0.1in}
\begin{itemize}
    \item  \textbf{Downward-Validity}: For all $i \in [1..n]$, $x_i \leq y_i$. 
    
    \item \textbf{Upward-Validity}: For all $i \in [1..n]$, $y_i \leq \sqcup\{x_1,...,x_n\}$.
    
    \item \textbf{Comparability}: For all $i \in [1..n]$ and $j 
\in [1..n]$, either $y_i \leq y_j$ or $y_j \leq y_i$.
\end{itemize}

\remove{

\subsection{The Byzantine Lattice Agreement Problem}

\lewis{As per discussion with Xiong, we will remove Byzantine part from this paper so that we have enough room to present main results and focus on the crash faults.}

The Byzantine lattice agreement problem (BLA)~\cite{di2019byzantine, zheng2019byzantineSync} is similar to the lattice agreement problem. Each node has an input value from a lattice and needs to output a value from the lattice such that all output values of correct processes are comparable. The difference lies on the validity condition. The \textbf{Upward-Validity} is modified as follows to prevent Byzantine nodes introducing arbitrary number of values into the outputs of correct nodes. 

\begin{itemize}
    \item \textbf{Upward-Validity}: $\sqcup \{y_i ~ | ~ i \in C\} \leq \sqcup(\{x_i ~ | ~ i \in C\} \cup B)$, where $B \subset X $ and $|B| \leq k$.  
\end{itemize}
}

%\paragraph*{Atomic Snapshot Object}
\vspace{3pt}
\noindent{\bf Atomic Snapshot Object (ASO)}:~~
The snapshot object is made up of $n$ segments (one per node), and provides two operations: \update{} and {\sc Scan}. Node $i$ invokes \update$(v)$ to write value $v$ into the $i$-th segment of the snapshot object. We adopt the single-writer semantics, i.e., only $i$ can write to the $i$-th segment.  The \scan operation allows a node to obtain an instantaneous view of the snapshot object. The \scan  returns a vector $Snap$, where $Snap[i]$ is a value of the $i$-th segment.   

%Let {\bf ASO} be  such  a  snapshot object. When  a  node $p_i$ invokes $update(v)$,  it  stores the value $v$ in its component {\bf SNAP}[$i$]. When a process $p_i$ invokes $scan()$, it obtains the value of all the components {\bf SNAP}[$1..n$]. The $scan$ operation allows a processor to obtain an instantaneous view of the snapshot object. 

%We define a partial order $\rightarrow$ on all \scan and \update{} operations in an execution such that $op_1 \rightarrow op_2$ if (and only if) the operation $op_1$ has terminated before the operation $op_2$ has started; that is, all low-level operations that comprise $op_1$ appear in the execution before any low-level operation that is part of $op_2$. The partial order $\rightarrow$ reflects the external real time order of non-overlapping operations in the execution.

Intuitively, a snapshot object is atomic  
(or linearizable)~\cite{herlihy1990linearizability} if every operation appears to happen instantaneously at some point in time between its invocation and response events. %To formalize it 
More formally, for each execution, there exists a sequence or ordering $\sigma$ that contains all \scan and \update{} operations in the execution and satisfies the following properties: %an implementation of the atomic snapshot object  %speakingsnapshot object implementation is atomic (linearizable)~\cite{herlihy1990linearizability}, if for 
\begin{itemize}
    \item \textbf{Real-time order:} If operation $Op_1$ completes before operation $Op_2$ starts in the execution, then $Op_1$ appears before $Op_2$ in $\sigma$.
    \item \textbf{Sequential specification:} If a \scan operation returns the vector $Snap$, then for every $i\in \{1 ,\cdots, n\}$, $Snap[i]$ is the value written by the \update{} operation by node $i$  that precedes the \scan operation in $\sigma$ or the initial value if no such \update{} exists.
\end{itemize}

\section{Early-Stopping Lattice  Agreement Algorithm}
\label{s:ELA}

\commentOut{++++++
\subsection{Model: Same as above}
\begin{enumerate}
    \item We do one shot lattice agreement
    \item Assume FIFO message delivery per node.
    \item Sequential computing at each node
    \item Assume Reliable Broadcast (See Bracha)
    
\end{enumerate}
++++++++}

%\subsection{Our Algorithm}

We first present \textit{ELA (Early-stopping Lattice Agreement)} in Algorithm \ref{algo:ELA}. 
Our constant amortized round atomic snapshot implementation uses a variation of ELA to ``order'' \scan and \update{}, which will be discussed in the next section. The ELA algorithm is inspired by the \textit{stable vector} algorithm by Attiya et al.~\cite{attiya1990renaming} and Mendes et al.~\cite{mendes2013topology}. One of our key contributions is the formal abstraction of the \textit{equivalence quorum} condition and its application to lattice agreement and atomic snapshot objects implementations.

Each node $i$
is given an input $x$, and at all times, $i$ maintains a vector of sets, $V_i[1]\cdots V_i[n]$, where $V_i[j]$ ($j\ne i$) stores the set of values received from node $j$. We denote that a variable $v$ belongs to node $i$ by attaching to it the subscript $i$, for example $v_i$. When the node identity is clear from the context, we often omit the subscript. 

% Each node has a main thread that executes the lattice agreement and other event-driven threads that handle the message receipt event, i.e., event handler. Lines 3 and 4 of the main thread, and each event handler is executed \textit{atomically}. The event handlers will be executed continuously even after the node decides an output $y$. 

%For any $q \not = p$, the element $V[j]$ stores the set of values received from node $j$. The element $V[i]$ stores all values received from any node\\ \hspace{1in}

%\vijay{strange name for the condition.} \lewis{Any good suggestion on the name? I named it Condition WAIT, because it tells a node how long it needs to wait for exposed values.} \xiong{I would suggest define it as a predicate, like $EquivalenceQurum$, also, it will be used in ASO}. \lewis{Yes, please go ahead. I'm not sure how you want to present your predicate.} 

The ELA algorithm has two main parts: exchange all values known so far, and determine when it is ``safe'' to output a value using a \textit{decision rule}.
One key challenge is to identify the decision rule to enable the early-stopping property. Our decision rule  is based on the existence of 
an \textit{equivalence quorum}. Let $V$ be a vector of size $n$ and $i\in \{1, \cdots,n\}$. We define the predicate $\EQ(V,i)$ as follows. 
\begin{definition}[Predicate $\EQ(V, i)$] 
\label{def:EQ}
$\EQ(V,i)$ is true iff $\exists Q \subseteq \{1, \cdots,n\}$ s.t. $ |Q| \geq n - f \wedge V[j] = V[i], \forall j \in Q$. When the predicate is true, we call  $Q$ as the \textbf{equivalence quorum}. 
\end{definition}

In ELA, node $i$ decides when the predicate \EQ$(V_i,i)$ becomes true \textit{for the first time}. Intuitively, 
$\EQ(V_i, i)$ becomes $true$ when node $i$ learns about $\ge n-f$ nodes (including $i$) with identical sets of values. %Once the predicate becomes true, 
Then, node $i$ decides on the join of all values in $V_i[i]$.

% Note that $V_i[i]$ may not be equal to $V_i[i]$ due to message delay. However, due to the FIFO and Reliable channel assumptions, we can guarantee that $V_i[i] \subseteq V_i[i]$. This together with the quorum intersection property ensure the correctness of our algorithm ELA.

%\lewis{+++++TODO+++++}

%\lewis{Need to discuss stable vector and early-stopping synchronous algorithm.}

%\vijay{declare types of the variables}

%\lewis{+++++TODO+++++}

\begin{algorithm*}[!hptb]
\begin{algorithmic}[1]
	\small
	%\footnotesize
	\item[{\bf Local Variables}: /* These variables can be accessed and modified by any thread at $i$. */]{}
    \item[] $x$ \COMMENT{input at node $i$}
    \item[] $V[1,\cdots, n]$ \COMMENT{vector of sets at node $i$, initially, $V[j] = \emptyset ~\forall j \not = i$, $V[i] = \{x\}$}

    % \item[\bf Initialization:]
    
    % \item[] $V \gets [\emptyset,\emptyset,\cdots,\{x\},\cdots,\emptyset]$,~~ with $\{x\}$ at position $i$
    
    \item[] \hrulefill

\begin{multicols}{2}
%\item[\textbf{Begin: Lattice Agreement}]
\item[{\bf When Lattice Agreement is invoked:}]
%\STATE 
%\STATE Add $x$ into $V[i]$ 
%\STATE \lewis{Added line above; otherwise, \wait{} will be satisfied with $V[i] = \emptyset$} \xiong{But $V[i]$ is initialized to include $x$}\lewis{Good point! Forgot about that.}
% \STATE $V \gets [\emptyset,\emptyset,\cdots,\{x\},\cdots,\emptyset]$,\\\hspace{0.6in}with $\{x\}$ at position $i$
\STATE  Send ($x,i$) to all 
\STATE \textbf{Wait until} $~~\EQ(V, i) = true$ \label{ELA-line: wait until n-f}
\STATE $V^* \gets$ the vector $V$ that satisfies\\ $\EQ(V, i)$ \underline{for the first time}
%\STATE  Wait until \textit{\wait{}} is satisfied \label{ELA-line: wait until n-f}
\STATE  \textbf{Decide} $y \gets \sqcup \{v ~| ~v \in V^*[i]\}$ \label{ELA-line: decide LA}

% \item[]
% \item[]
\item[/* Event handler: executing atomically in]
\item[background, even after node $i$ decides */]
% \item[/* ]

% \item[/* Lines 6 to 8 executed atomically */] 

\item[{\bf Upon receiving $(v, j)$ from node $j$}:]

\STATE Add $v$ into $V[j]$, $V[i]$ \label{line:ela_add}

\IF{$(v,i)$ has \textit{not} been sent before}

\STATE Send $(v, i)$ to all
%Send $v$ to all
\label{ELA-line:echo to all}
\ENDIF

%\sapta{If we do this, we lose the notion of which value came from which node}
%\STATE  \hspace{0.2in} Bcast($y,q, p$)

% \item[]
% \item[]

%\STATE{\bf Upon receiving ($y,q,r$)}
%\STATE add $y$ to position $j$ in $V[r]$
%\STATE {\bf Bcast($y,q, p$)}
\end{multicols}
\end{algorithmic}
\caption{ELA (Early-stopping Lattice Agreement): Code for node $i$}
\label{algo:ELA}
\end{algorithm*}
% \xiong{Can we remove the space in Algorithm 1?}

% \vijay{Add that the algorithm satisfies the invariant that for all processes $i$: $V_i[j] \subseteq V_i[i]$.} 

%We say that vectors $V_i \subseteq V_j$ iff for all $i= 1,2,\cdots n$, $V_i[i] \subseteq V_j[i]$.

%\xiong{we don't need the definition of comparable vectors, we only say comparable sets} \lewis{Good point. Comment that out.}

%\begin{definition}[Comparable Vectors]
%We say that any two vectors $V_i$ and $V_j$ at nodes $i$ and $j$ are \textit{comparable} iff $V_i \subseteq V_j$ or $V_j\subseteq V_i$.
%\end{definition}

%\paragraph*{Correctness of ELA}
% \vspace{1pt}
\noindent{\bf Correctness of ELA}:~~
Consider any execution of Algorithm~\ref{algo:ELA}. We show that the outputs of correct nodes satisfy the three properties defined in Section~\ref{s:intro}. Due to space constraint, proofs are presented in Appendix \ref{app:ELA}. 
Downward-validity and upward-validity are straightforward from the code. 
Lemmas~\ref{lem:comparable_views_for_same_node} is key for proving comparability in Lemma~\ref{lem:comparability}.
For any two sets $U$ and $V$, we say $U$ and $V$ are \textit{comparable} if either $U \subseteq V$ or $V \subseteq U$.
%\begin{enumerate}
%    \item Define round
%\end{enumerate}

% \xiong{The proof of following lemma is not clear. Maybe we should use $V_i^{t}[s]$ to denote the set $V_i^t [s]$ at time $t$?} \xiong{A simple proof would be as follows. Since communication is FIFO and the history of node $s$ is increasing, the set $V_i[s]$ at time $t$ must be the same as the history of node $s$ at some time. The same thing holds for $V_j[s]$ at time $t'$. Thus, they are comparable.}

%\lewis{value\_at\_time($V_i[s], t$)}

%\lewis{V_{p(t)}[s]}

%\lewis{$V_i[s](t)$}

%\lewis{``$V_j[s]$ at any time $t$''}

\begin{restatable}{lemma}{compViewsForSameNode} \label{lem:comparable_views_for_same_node}
For any two nodes $i$ and $j$, fix time $t$ and $t'$, and then the set $V_i[s]$ at time $t$ and the set $V_j[s]$ at time $t'$ are comparable for each node $s$.
\end{restatable}

% \begin{proof}
% The value of set $V_i[s]$ is modified only when $i$ receives a message from node $s$. Since $s$ is a non-faulty node and the communication is FIFO, the set $V_i[s]$ at time $t$ must be the same as the set $V_s[s]$ at some time $t_i$ and the set $V_j[s]$ at time $t'$ must be the same as the set $V_s[s]$ at some time $t_j$. The set $V_s[s]$ is non-decreasing. Thus, $V_i[s]$ at time $t$ must be comparable with $V_j[s]$ at time $t'$. 
% % For the sake of contradiction, suppose there exists a value $v$ in $V_i[s]$ at time $t$ but not in $V_j[s]$ at time $t'$ and a value $u$ in $V_j[s]$ at time $t'$ but not in  $V_i[s]$ at time $t$. For node $s$, it either sends value $v$ before $u$ or $u$ before $v$. 
% % Suppose $s$ sends $v$ before $u$, w.l.o.g. Because the communication channel is FIFO, both node $i$ and node $j$ must receive value $v$ before $u$, which contradicts the fact that $v \in V_i[s] - V_j[s]$. Thus, $V_i[s]$ and $V_j[s]$ must be comparable. 
% \end{proof}

By applying Lemma \ref{lem:comparable_views_for_same_node} and the decision rule, we have the following lemma. 
\begin{restatable}{lemma}{comparabilityELA}[Comparability] \label{lem:comparability}
For any two nodes $i$ and $j$, $y_i$ and $y_j$ are comparable.  
\end{restatable}
% \begin{proof}
% Let $V_i$ denote the vector at node $i$ and $V_j$ denote the vector at node $j$ when nodes $i$ and $j$ decide. The statement of the lemma is proved if we show that $V_i[i] $ and $V_j[j]$ are comparable. 

% The decision condition on line \ref{ELA-line: wait until n-f} states that there exists a set $Q_i$ of size at least $n - f$ such that $V_i[i] = V_i[s]$ for each $s \in Q_i$ and a set $Q_j$ of size at least $n - f$ such that $V_j[j] = V_j[s]$ for each $s \in Q_j$. Since $f < \frac{n}{2}$, there exists a correct process $s \in Q_i \cap Q_j$. Lemma~\ref{lem:comparable_views_for_same_node} implies that $V_i[s]$ and $V_j[s]$ are comparable. This leads to the conclusion that $V_i[i]$ (which is equal to $V_i[s]$) is comparable to $V_j[j]$ (which is equal to $V_j[s]$)
% %Applying Lemma~\ref{lem:comparable_views_for_same_node} we can see that at some time $t$, $V_i[i] = V_i[s]$ and at some time $t'$ $V_j[j] = V_j[s]$
% \end{proof}

%\paragraph*{Round Complexity}
\vspace{3pt}
\noindent{\bf Round Complexity}:~~
Given an execution of ELA, let $D$ be the \textit{maximum message delay}. That is, if both the sender and the receiver are nonfaulty, then the sender's message will be received by the receiver within time $D$. We divide time into intervals of length $D$ and each interval is called a \textit{round}. 
For simplicity, we assume that every node initiates the ELA algorithm at the same time. %\footnote{This assumption can be removed by having the nodes start at arbitrary times and then counting the maximum message delay differently. Let $t$ be the time when the last non faulty node sends out its first message. The new \textit{maximum message delay} is defined as $t+D$. } 
% where $D'$ is time when the last ``first'' message from a correct node to reach another non-faulty node. }.
%\sapta{***cite some papers that have same/similar assumptions**} 
The analysis can be generalized to the case when nodes invoke ELA within constant number of rounds. % \sapta{<-what does this mean? }.
We begin with a useful definition. 

\begin{definition}
% \vijay{Need another name for "New" value -- confusing otherwise. Some choices are "revealed" value or "exposed" value}
[Exposed value in an interval] \label{def:new_value}
We say a value $v$ is an exposed value in interval $[t, t + D)$ if some nonfaulty node receives $v$ in interval $[t, t+D)$, and \underline{no nonfaulty node} has received $v$ before time $t$. 
\end{definition}

Note that by definition, any exposed value for $t > D$ must be the input of some faulty node. We have the following lemma, which guarantees the termination of our algorithm.
% Recall that $C$ denotes the set of correct nodes in the execution. \lewis{If we don't use $C$ often, I'd suggest removing it altogether.}

\begin{restatable}{lemma}{terminationELA}[Termination] \label{lem:termination_without_newValue}
For an arbitrary interval $[t, t + 2D)$. If there does not exist any exposed value in this interval, then all undecided nonfaulty nodes decide by time $t + 2D$. 
\end{restatable}

Lemma \ref{lem:termination_without_newValue} follows from the observation that if there is \textit{no} exposed value in the $[t, t + 2D)$ interval, then at the end of the interval, for each node $i$, we must have $V_i[j] = V_i[i]$ for each nonfaulty $j$. Now we introduce the notion of \textit{failure chain} of an exposed value.

\begin{definition}
[Failure chain of an exposed value]
\label{def:failure_chain}
A sequence of nodes $p_1,p_2,...,p_m$ is said to form a failure chain of an exposed  value $v$ if (i) $p_1, p_2, \dots, p_{m -1}$ are faulty, and $p_m$ is correct; (ii) the input value of $p_1$ is $v$; (iii) $p_i$ receives value $v$ from $p_{i-1}$; and (iv) For $1 \leq i < m - 1$, $p_i$ crashes while sending $(v,p_i)$ to other nodes, i.e., $p_1$ crashes when executing line 1 and $p_2,...,p_{m-2}$ crash when executing line \ref{ELA-line:echo to all}. 
\end{definition}

% \lewis{The intervals below are $(]$ and $[]$, but the definition above is $[)$. Some reviewers will complain.}

\begin{restatable}{lemma}{failureChainLength} \label{lem:failureChainLength}
If value $v$ is an exposed value in interval $[t, t + D)$, then value $v$ has a failure chain with length at least $\frac{t}{D} + 1$.
\end{restatable}

The following lemma can be derived from condition $(iv)$ of Definition \ref{def:failure_chain}. 

\begin{restatable}{lemma}{uniqueFailureChain} \label{lem:unique_failure_chain}
For any two exposed values $v$ and $u$ with failure chain $P_v$ and $P_u$ respectively. Then, the first $|P_v| - 2$ nodes in $P_v$ and the first $|P_u| - 2$ nodes in $P_u$ are disjoint.
% A faulty node cannot be in the failure chains of two different exposed values. 
\end{restatable}

% \begin{proof}
% Let $V$ and $U$ denote the set of the first $|P_v| - 2$ nodes in $P_v$ and the first $|P_u| - 2$ nodes in $P_u$. Suppose node $i \in V \cap U$ for contradiction. By condition (iv) of Definition \ref{def:failure_chain}, node $i$ crashes while sending $v$ to other nodes on line~\ref{ELA-line:echo to all} of ELA. Since lines \ref{line:ela_add} to~\ref{ELA-line:echo to all} of ELA are executed atomically, node $i$ cannot crash while sending value $u$ to other nodes at line~\ref{ELA-line:echo to all}, a contradiction. 
% \end{proof}

\begin{restatable}{lemma}{roundComplexityELA} \label{lem:round_complexity_single_shot}
If an execution has $k\le f$ crash failures, then ELA takes at most $2 \sqrt{k}$ rounds. 
\end{restatable}
\begin{proof}[Proof Sketch]
Since there are at most $k$ failures in the execution, and an exposed value in interval $[t, t + D)$ is associated with $\geq \frac{t}{D} - 1$ unique faulty nodes by Lemma \ref{lem:failureChainLength} and \ref{lem:unique_failure_chain}, we cannot have $> 2\sqrt{k}$ distinct intervals with exposed values. Lemma \ref{lem:termination_without_newValue} then implies that ELA takes at most $2\sqrt{k}$ rounds.  
% Let us assume that the algorithm takes $ 2 \sqrt{k} + 1 $ rounds. By Lemma \ref{lem:termination_without_newValue},  we know that to prevent the algorithm from terminating, there has to be at least one exposed value every two rounds. Lemma~\ref{lem:failureChainLength} gives us the length of any failure chain and Lemma~\ref{lem:unique_failure_chain} states that a faulty node (except for the last 2 nodes in a failure chain) can be a part of only one failure chain.
% Thus if the algorithm terminates in round $2 \sqrt{k} + 1 $,  the number of faulty nodes must be at least $1 + 3 + \dots + 2\sqrt{k} - 1 > k$, leading to a contradiction. 
\end{proof}

\section{Atomic Snapshot Object}
\label{s:aso}

% \sapta{I think def of History should come before the following definition}
% \xiong{I think we don't need this definition}

% \begin{definition}
% Two histories $V$ and $V'$ are said to be \underline{equal up to tag $t$} if
% \sapta{timestamp = tag, id. so it should be $a\leq t$ and I think we should add $\forall (a,*): a\le t$ inside the next definition}
% $$
% \{ (x, (a, b))~~|~~(x, (a, b)) \in V, b \leq t\} = \{ (x, (a, b)) ~~|~~(x, (a, b)) \in V', b \leq t\}.
% $$
% \end{definition}

In this section, we present two algorithms for implementing an atomic snapshot object in crash-prone asynchronous message-passing systems with $f < \frac{n}{2}$. 

%%%%%% \lewis{Remove the paragraph below, since it's a bit redundant and we're short of space.}
%First, we show how to adapt the atomic snapshot algorithm in shared memory by Attiya et al.~\cite{attiya1995atomic} for the message-passing systems. The algorithm in~\cite{attiya1995atomic} is  a general transformation that uses a lattice agreement algorithm as a blackbox to implement an atomic snapshot.  Our algorithm is also general in this sense. Using the $O(\log n)$-round lattice agreement algorithm  in~\cite{zheng2018linearizable} inside our adapted transformation, we obtain an atomic snapshot object that has $O(\log n)$ rounds for both \update{} and \scan operations in the message-passing systems. 
%Second, we present an algorithm that has \textit{constant amortized} round complexity for both \update{} and \scan operations. This algorithm applies the equivalence quorum technique to directly implement an atomic snapshot object. 

%Each operation takes $O(\sqrt{k})$ message delays in the worst case and takes constant message delays in the amortized sense, even in the presence of crash failures. The message size is also constant. \lewis{What do we mean by constant? Maybe the overhead is constant? Message size is the same as the input size, right?} \xiong{I mean if the input value has constant message size, then each message has constant size. In some sense, the overhead is constant. } 

\subsection{General Transformation}
 Attiya et al.~\cite{attiya1995atomic} gave an elegant algorithm that transforms any wait-free lattice agreement algorithm to a wait-free atomic snapshot object in the shared memory systems. Their key idea is to invoke a sequence of lattice agreement instances to obtain \textit{comparable} snapshots. 
% Their algorithm is  a general transformation that uses a lattice agreement algorithm as a blackbox to implement an atomic snapshot.  Our algorithm is also general in this sense.
%The technique can be viewed as a variation of double-collect. \lewis{Xiong: please add some citations. Who is the first one to use double-collect?} \xiong{I don't think we need to say this. It's not exactly double-collect, they are using lattice agreement} \lewis{OK. Noted. I will address this part tomorrow.}
%of the object.\footnote{In the terminology of~\cite{attiya1995atomic}, a view is a snapshot. We will slightly abuse the term ``view'' to facilitate our discussion. Different notions of views are introduced in Definition \ref{def:view}. \lewis{We might or might not need this footnote. I'm still debating. The reason I have it here is that I'm worried that if someone just assumes they know what view is and skips our definition. Any thoughts?} \xiong{I think its ok not to have this. }}

%By the comparability of lattice agreement, the operations that invoke the same instance of lattice agreement must obtain comparable views. So the difficulty lies on how to ensure that operations invoking different lattice agreement instances obtain comparable views. 

To adapt the algorithm in~\cite{attiya1995atomic} for  message-passing systems, we need to make two main modifications: (i) Replace each read or write step in shared memory by sending a read or write message to all nodes and waiting for $n - f$ acknowledgements; and (ii) Add another write step (sending the input to at least $n - f$ nodes) before invoking a lattice agreement instance. Since the algorithm is similar to the algorithm in~\cite{attiya1995atomic} except for these two changes,  we present the algorithm, TS-ASO, and its proof in Appendix \ref{app:transformation_message}.
Using the $O(\log n)$-round lattice agreement algorithm by Xiong et al.~\cite{zheng2018linearizable} in our transformation gives an implementation of atomic snapshot objects that take $O(\log n)$ rounds for both \update{} and \scan operations.\footnote{Note that the algorithm in~\cite{zheng2018linearizable} actually has round complexity $O(\log f)$; however, if we plug in the original version, our transformation becomes $O(n)$ rounds. We need to make a simple modification of the algorithm in~\cite{zheng2018linearizable} to get $O(\log n)$ round complexity. Please refer to Appendix \ref{app:transformation_message} for more details.} 

% The proof is rather mechanical and follows from the proof structure of ~\cite{attiya1995atomic}; hence, 

%as described below. 
%\begin{itemize}
%    \item Replace each read or write step to the shared memory by sending a read or write message to all nodes and waiting for $n - f$ acknowledgements. This ensures that our transformation algorithm satisfies liveness. \xiong{not just for liveness}
    
%    \item Adding another write step before an operation invoking a lattice agreement instance. \lewis{Xiong: please add a short comment on why this is necessary. Maybe making sure that the view won't be lost?} \xiong{This step is to ensure that lattice agreement with higher tag would observe all the inputs for lattice agreement with lower tag, thus has an input which is bigger than the output of any lattice agreement with lower tag}
%\end{itemize}

%The first one is to replace each read or write step with sending a read or write message to all nodes and waiting for $n - f$ acknowledgements. The second one is to add another write step before an operation invoking a lattice agreement instance. Due to space limitation, we present the algorithm and its full proof in Appendix \ref{app:transformation_message}. 

%Although the above algorithm transforms any lattice agreement algorithm to atomic snapshot object implementation in message-passing systems, 

%Two examples of such algorithms are the $O(\log f)$ rounds algorithm by Xiong et al.~\cite{zheng2018linearizable} and our ELA algorithm presented in Section \ref{s:ELA}. 

One drawback of our transformation is that it does \textit{not} necessarily ``preserve'' the round complexity of the lattice agreement algorithm. This is because the round complexity analysis of some lattice agreement algorithms depends on the assumption that each node starts around the same time and different nodes might participate in the same lattice agreement instance at different times in TS-ASO. Therefore, directly using our ELA algorithm in the transformation gives a round complexity of $O(n)$.\footnote{The round complexity guarantee of the $O(\log n)$ rounds algorithm in~\cite{zheng2018linearizable} does not depend on the assumption that all nodes start the algorithm around the same time.} To address this issue, we propose our second atomic snapshot algorithm. %Recall that our ELA algorithm presented in Section \ref{s:ELA} relies on the assumption that each node invokes the lattice agreement instance around the same time, whereas the round complexity of the algorithm in~\cite{zheng2018linearizable} does not have such an assumption.  

\subsection{Algorithm with Constant Amortized Round Complexity}

Our second atomic snapshot algorithm, \alg{} (amortized constant atomic snapshot object), uses the equivalence quorum technique and a novel mechanism of invoking lattice agreement instances
to ensure amortized round complexity.
In addition, TS-ASO requires \textit{message size overhead} of $O(n)$, because each node needs to collect the states of at least a quorum of nodes before participating in a particular lattice agreement instance. \alg{} only incurs \textit{$O(1)$ message size overhead}. As a byproduct, we obtain a linearizable update-query state machines~\cite{faleiro2012generalized} that take amortized $O(1)$ rounds for each update and query command and $O(1)$ message size overhead, shown in Appendix \ref{app:linearizable_UQ}.

%in a \scan or \update{} operation.
%\lewis{I removed discussion on Q/U, since it's a bit distracting for reviewers that have no idea about what Q/U is. It's not as common.} \xiong{Yes, we don't need to discuss over here. We have a separate section for U/Q state machine}
%Worse, for applications in update/query state machines, the message size becomes unbounded, since a message can contain all commands in history of the state machine. In contrast, our lattice agreement algorithm in previous section only require constant message size. We would like to propose some algorithm for atomic snapshot object which also requires constant message size. 

%We propose a new technique to address the issues. \lewis{XXXXX Need to add something here later. XXXXXX} 
% Note that our algorithm \alg{} is \textit{not} a general transformation from lattice agreement to atomic snapshot. 
% The way we invoke lattice agreement is tightly integrated with the other parts of the algorithm.
% A transformation algorithm that preserves the round complexity of lattice agreement is an interesting topic on its own and will be studied as a future work. \lewis{Maybe we don't need this paragraph if we don't have enough space.}

\subsubsection{Main Techniques of \alg{}}
\label{s:aso-overview}

Our algorithm \alg{} is inspired by \cite{attiya1995atomic}, i.e., invoking a sequence of lattice agreement instances to implement atomic snapshot. The key technical contribution is to identify how to tightly glue different components together to obtain amortized constant round complexity. We first discuss two goals that need to be achieved for  correctness. Then we introduce a new mechanism of invoking lattice operations, namely $LatticeRenewal()$, and discuss how we achieve the desired round complexity. Finally, we compare TS-ASO and AC-ASO.

In our discussion below, we call an instance of lattice agreement a \textit{lattice operation} for brevity.
Following \cite{attiya1995atomic}, we will use a tag (or a logical timestamp) to distinguish different lattice operations in \alg. Hence by ``nodes participate in \underline{lattice operation with the same tag},'' we mean that these nodes are in the same instance of the lattice agreement algorithm. Due to the property of lattice agreement, these nodes are guaranteed to obtain comparable outputs.
Each value written by the \update{} is also assigned a tag as well. Later in Definition \ref{def:tag_operation}, we  formally define the tag for operations and values.

%\paragraph*{Safety}

\vspace{5pt}
\noindent{\bf Goals for ensuring correctness}:~~
%We first describe how \alg{} achieves safety and liveness. 
In our design, when a \scan or \update{} operation completes, it obtains a ``view.''
Roughly speaking, a view represents a set of values that are observed by the operation and are ``safe'' to return (to be introduced formally later in Definition \ref{def:view}).
Inspired by \cite{attiya1995atomic}, we want to achieve the following two goals in our algorithm: (G1) views obtained by all \scan and \update{} operations are \textit{comparable}; and (G2) once an \update{} completes, its written value is ``visible'' to any subsequent operations. These two goals allow us to use a natural mechanism to construct a linearization for a given execution. 

It is simple to achieve goal (G2). \alg{} ensures that once an \update{} completes, a quorum of nodes have seen the update. 
For goal (G1), we require each node to participate in lattice operation(s) to complete its \scan and \update{} operations.  
Intuitively, an operation is completed if it obtains a view, which could be an output of a lattice operation invoked by this operation, or an output \textit{borrowed} from another lattice operation invoked by some other operation. 
\alg{} achieves goal (G1) by maintaining the following invariant:

%On a high-level, let ``views with tag at most $r$$'' denote a set of values with tag at most $r$ (to be defined formally later in Definition \ref{def:view}). Then our algorithm \alg  enforces the following invariant:

\begin{invariant}
\label{invariant}
In \alg, the views returned by lattice operations are \textit{comparable}.
\end{invariant}

\vspace{5pt}
\noindent\textbf{Lattice Renewal:}~~
\label{sec:LR}
%Inspired by \cite{attiya1995atomic}, 
We introduce the $LatticeRenewal()$ procedure to guarantee Invariant \ref{invariant}.
We stress that even though the usage of borrowed view is not new, we are not aware of any prior work that achieves amortized constant rounds for atomic snapshot objects. 
$LatticeRenewal()$ is a mechanism to invoke a sequence of lattice operations to provide the following desirable properties:

\begin{itemize}

    \item[(P1)] $LatticeRenewal()$ invokes at most three lattice operations in a row.
    
    % Lewis: Remove this according to Xiong's suggestion. This property is not precise. Nodes may actually join much later.
    %\item Each node that participates in the lattice operation with the same tag around the same time (i.e., at most constant rounds apart) 
    
    %\xiong{We should put this in round complexity part, we need this because of preserving the round complexity} \lewis{Let me think about it. A benefit of presenting this here is to have a full picture of LR, but I also see your point. Maybe I can move this to the last, and mention that this is important for improved round complexity?} \xiong{Let's keep it this way, since you have some explanation which property is for round complexity}
    
    \item[(P2)] If any of the lattice operations does \textit{not} observe a higher tag, then $LatticeRenewal()$ returns the view obtained by that particular lattice operation, namely \textit{direct view}.
    
    \item[(P3)] If all three lattice operations observe a larger tag, then $LatticeRenewal()$ fails to find a direct view. It  will then wait to \textit{borrow} a view from a lattice operation invoked by other nodes, namely \textit{indirect view}. 
    
    \item[(P4)] Views returned by $LatticeRenewal()$ are comparable with each other.
% \end{enumerate}
 \end{itemize}

%To achieve property 2, we make sure that each lattice operation with tag $T$ only ``observes'' the values with tag at most $T$.  This allows the lattice operation to ``ignore'' nodes that participate too late.

% XXXX Lewis: We might add this paragraph back if we have space. XXXX
%Our usage of timestamp ensures that after a lattice operation with tag $T$ started, all values from \update{} operations that start in some constant time \textit{later} must obtain tags strictly greater than $T$. This help \alg{} achieves property 2, since using the tag, a lattice operation can ``ignore'' nodes that participate too late.
% XXXX
%This implies that all possible new values must be values with tag at most $T$. These ``late'' values would be ignored by the lattice operation with tag $T$, because the particular instance of lattice agreement is only with respect to the values with tag at most $T$.
%all values that should be ``observed'' in $LO(T)$ must have tag at most $T$.

%Another trick we use is to have each \update{} operation send outs its value to all the other nodes before invoking its initial lattice operation. This trick ensures that there can be at most $O(k)$ new values for the lattice operation after a constant number of rounds. \lewis{Xiong: This sentence is not very clear why it's important.}

(P1) is mainly for correctness and improved round complexity as we will explain next. (P2) to (P4) jointly guarantee Invariant \ref{invariant}. Due to the properties of lattice agreement, views returned by the lattice operation with the same tag are comparable. For lattice operations with different tags, we rely on (P2) and (P3). (P2) implies that a lattice operation returns a view iff it does \textit{not} observe a lattice operation with a higher tag. 
This together with our approach of obtaining tags ensure that the view returned by a lattice operation with a smaller tag must be known by a lattice operation with a larger tag.
Therefore, when the lattice operation with a larger tag starts, its view is at least as large as the view of any lattice operation with a smaller tag.
This allows later lattice operations to learn older views.

Due to message delays and concurrent \update{}s, it is possible that all three lattice operations fail to return a view. In this case, we rely on (P3) to ensure that $LatticeRenewal()$ is able to obtain an indirect view. Moreover, our design guarantees that such an indirect view can be borrowed within a constant amount of rounds. In Lemma \ref{lem:lattice_comparable_view}, we formally prove that properties (P2) to (P4) are enough to maintain Invariant \ref{invariant}.

In \alg{}, a \scan and \update{} operation invokes $LatticeRenewal()$ (after some preprocessing) and the operation is completed when $LatticeRenewal()$ obtains a view. Invariant \ref{invariant} can be used to prove (P4), which then guarantees goal (G1) -- views obtained by all \scan and \update{} operations are \textit{comparable} (as formally proved in Lemma \ref{lem:comparable_views}). 
Later in our correctness proof, this allows us to construct a linearization of operations.
 
% Lewis: It is tricky to explain on why we need three LOs. If we don't find a good way, we might not discuss this; otherwise, it gets too technical details too early.
%To obtain a view, \alg{} relies on \textit{three} consecutive lattice operations, whereas the algorithm in \cite{attiya1995atomic} (and our transformation in Appendix XXX) only uses a read step and \textit{two} lattice operations. This subtle difference is the key component that allows us to achieve better round complexity.\xiong{The three phase is needed for correctness, i.e., for goal (i) and (ii). The reason why we need three instead of two because of goal (ii)} \lewis{We might want to discuss this in more details in Appendix, e.g., why \alg{} cannot use ``a read + 2 LOs''.} Moreover, we need three lattice operations to ensure comparable views between lattice operations with different tags.

\vspace{5pt}
\noindent{\bf Round Complexity:}~~
\alg{} ensures amortized constant round complexity, and each operation takes $O(\sqrt{k})$ rounds in the worst case. 
On a high level, \alg{} uses the \textit{equivalence quorum} technique to implement the underlying lattice operation. %\xiong{I wouldn't say use ELA for the lattice operation, we should say it uses the equivalence quorum} 
Worst case round complexity roughly follows the analysis for ELA (Algorithm \ref{algo:ELA}) as presented in Section \ref{s:ELA}.
For amortized constant round, the main property we rely on is the \textit{early-stopping} property which ensures that \textit {if no node fails, then the lattice operation completes in a constant number of rounds}.  By assumption, a crashed node does not participate in the algorithm anymore; hence, if we have enough number of \scan and \update{} operations, \alg{} achieves amortized constant round complexity. 

Recall that the round complexity analysis of ELA depends on the notion of \textit{exposed values} (Definition \ref{def:new_value}) and the time interval these values appear. Unlike the (single-shot) lattice agreement problem, atomic snapshot is long-living; thus, it is possible that \update{}s attempt to write values with the same tag consecutively in a way that these values are all treated as the input to a particular lattice operation, which eventually ``slow down'' the progress of that lattice operation, and hence the \update{} operation. This is also the reason that if we plug ELA into our transformation algorithm in Appendix \ref{app:transformation_message}, we get $O(n)$ round complexity.

Our idea to address this issue is in fact simple: \textit{increment tags and invoke lattice operation(s) in a way that ``late'' \update{} \underline{does not prevent the progress} of existing
lattice operations.} More concretely, 
consider a lattice operation with tag $T$, say $Lattice(T)$, starts at time $t$. Fix a constant $D$ (which will become clear in Lemma \ref{lem:later_operation_bigger_tag}). We need to ensure that (i) All values from \update{} operations that start after time $t + D$ must have tag strictly greater than $T$; and (ii) ``slow writers'' that participate in $Lattice(T)$ after time $t+D$ do \textit{not} introduce any exposed value with tag $T$. 
These two properties ensure that the lattice operation in \alg{} completes in $O(\sqrt{k})$ rounds in the worst case, and has a constant amortized round complexity. This observation together with the design that $LatticeRenewal()$ invokes at most three lattice operations give the desired round complexity.

\vspace{3pt}
\noindent\textbf{TS-ASO vs AC-ASO}:~~
%Our intuitions behind \alg{} are simple and
% Our algorithm structure is inspired by the transformation algorithm in \cite{attiya1995atomic}, but the challenges lie in gluing different components together and our approaches are often subtle. 
Recall that TS-ASO is our general transformation algorithm adapted from \cite{attiya1995atomic} (presented in Appendix \ref{app:transformation_message}). We present high-level comparison between \alg{} and TS-ASO here, and details in Appendix \ref{app:comparison_TS-ASO_AC-ASO}. 
% The key difference lies in re-design the \update{} operation improve the round and message overhead complexity, we need to  as follows:
\begin{itemize}
    \item[(D1)] TS-ASO participates in the first lattice operation using the largest tag read from a quorum, whereas \alg{} adds $1$ to obtain a new tag for the first lattice operation.
    \item[(D2)] \alg{} has an initial lattice operation in addition to the ones in $LatticeRenewal()$.
    \item[(D3)] TS-ASO collects the states of at least a quorum and writes the join of the states collected to at least a quorum, and participates in at most two lattice operations to obtain a view. \alg{}, instead, directly uses at most three lattice operations to obtain a view.
\end{itemize}

Roughly speaking, (D1) allows nodes to invoke a lattice operation with the same tag around the same time. (D2) ensures property (P3) of $LatticeRenewal()$, particularly, some node without obtaining a direct view can always borrow a view. (D3) ensures property (P4) of $LatticeRenewal()$ and constant message size overhead.
The details of the necessity of three lattice operations are presented in Lemma \ref{lem:output_view_dominate_input_view}, and the round complexity analysis can be found in Lemma \ref{lem:worst_round_complexity_multishot}. 

\begin{algorithm*}[!hptb]
\begin{algorithmic}[1]
	\small
	%\footnotesize
	\item[{\bf Local Variables:} /* These variables can be accessed and modified by any thread at $p$. */]{}
	
    \item[] $V[1\cdots n]$ \COMMENT{vector of ``views''. $V[j]$ is the set of values received from $j$}
    
    \item[] $maxTag$ \COMMENT{integer, largest tag ever seen via $\quotes{writeTag}$, $\quotes{echoTag}$ messages.} \\
    % \item[] $R$ \COMMENT{Integer, local tag for lattice agreement} 
    
    \item[] $D[1 \cdots n]$ \COMMENT{vector of views from \underline{good lattice operations}.}
    
    \vspace{5pt}
    
    \item[{\bf Derived Variable:}]
    
    \item[] $V^{\leq r} \gets [V[1]^{\leq r}, V[2]^{\leq r}, \dots, \dots, V[n]^{\leq r}]$
    \COMMENT{vector of ``views'' w/ tag at most $r$}
    
    \item[] \hrulefill

\begin{multicols}{2}
\item[{\bf Initialization}:]
\STATE $V \gets [\emptyset,\emptyset,\cdots,\cdots,\emptyset]$
\STATE $D \gets [\emptyset,\emptyset,\cdots,\cdots,\emptyset]$
\STATE $maxTag \gets 0$
\item[]

\item[{\bf When \update$(v)$ is invoked}:]
%\item[\textbf{Begin: Update$(v)$}]
%\STATE 
\STATE $r \gets readTag()$ \label{line:update_readTag}
% \STATE $writeTag(r)$
% \STATE Wait until at least $n-f$ distinct $V[q]$'s are  \\ \hspace{0.7in} equal to $V[p]$ up to tag $r$; 
% \STATE {\bf if} $maxTag = r$ 
% \STATE \h Send $(``completed", r)$ to all
\STATE $ts \gets ~ \tuple{r + 1, i}$ \label{line:value_ts}

\STATE Send $(\quotes{value},  \tuple{v, ts})$ to all \label{line:brodcast_val}
\STATE Lattice$(r)$~~~~~~~~~~~~$\rhd$\texttt{Phase 0} \label{line:initial_phase}%/* Initial Phase*/ 
\STATE $r' \gets \max\{r + 1, maxTag\}$ \label{line:tag_increment}
\STATE $updateView \gets $ LatticeRenewal$(r')$ \label{line:update_lattice_renewals}
\STATE \textbf{Return} ACK

\item[]
\item[{\bf When {\sc Scan}() is invoked}:]
%\item[\textbf{Begin: Scan}]
\STATE $r \gets readTag()$ \label{line:scan_readTag}
\STATE $scanView \gets$ LatticeRenewal$(r)$ \label{line:scan_lattice_renewals}
\STATE \textbf{Return} $extract(scanView)$

\item[]
\item[
/* Helper procedures */]
\vspace{5pt}
%\item[\textbf{Begin: Lattice($r$)}]
\item[\bf Procedure Lattice($r$):]
\STATE $writeTag(r)$  \label{line:lattice_writeTag}
% \xiong{need to define $V^{\leq r}$} \lewis{Added derived variable above.}
\STATE \textbf{Wait until} $~~\EQ(V^{\leq r}, i) = True$
\label{line:lattice_termination_condition}   %\sapta{Wait until history is verified upto tag $R$.} 
\\/* Execute lines \ref{line:eq_true} to line \ref{line:lattice_end} atomically */ %\xiong{There is a problem. It should be lines 16-20 are executed atomically when EQ is true} \lewis{Yes, it was not precise. How about the new change?}
\STATE $V^* \gets$ the vector $V^{\leq r}$ that satisfies\\  $\EQ(V^{\leq r}, i)$ \underline{for the first time} \label{line:eq_true}
\IF{$maxTag \leq r$} \label{line:good_lattice_condition}

\STATE Send $(``goodLA", r)$ to all \label{line:broadcast_decision}
%Lewis: change completed to completeLA
% \STATE \h {\bf Return} $(True, extractView(V[p]^{\leq R}))$ \label{line:lattice_return_true}
\STATE {\bf Return} %$(true, V[p]^{\leq r})$ 
$(true, V^*[i])$
\label{line:lattice_return_true}
\ELSE
\STATE {\bf Return} $(false, \emptyset)$ \label{line:lattice_end} \ENDIF
\item[]

\item[\textbf{Procedure LatticeRenewal($r$)}:]
\FOR{$phase \gets 1$ to $3$}
\STATE $(status, view) \gets Lattice(r)$ 
\IF{$status = true$}
\STATE{\bf Return} $view$~~~~~~~$\rhd$\texttt{Direct View}
\label{line:direct}
\ELSIF{$phase = 3$}
\STATE {\bf Break} \label{line:bad_lattice_op}
\ENDIF
\STATE $r \gets maxTag$ 
%\lewis{Do we need this line? We can have $(\quotes{completed}, maxTag)$ below? Another thing is that, after line 29, we need to run line 46 first? Maybe we need to add another predicate here? Say wait until line 46 is executed?}
\ENDFOR
% \item[] /* Lines 29 \& 30 executed atomically */ \xiong{No need}
\STATE
\textbf{Wait until} receiving $(\quotes{goodLA}, r)$ 
\\\hspace{0.8in}from some node $j$ \label{line:wait_borrow}
% \STATE {\bf Return} $extractView(D[q])$ \label{line:indirect}
\STATE {\bf Return} $D[j]$~~~~~~~~~~$\rhd$\texttt{Indirect View} \label{line:indirect}
%\item[\textbf{End: LatticeRenewals}]
\item[]

\item[\textbf{{Procedure extract($S$)}}:]
\STATE $Snap \gets [1, \cdots, n]$
\FOR{$j = 1$ to $n$}
\STATE $Snap[j] \gets v$, where $\langle v, \langle t', j \rangle \rangle \in S$, and
\\ \hspace{0.3in} $t'$ is the largest tag of $j$'s values in $S$
\ENDFOR
\STATE {\bf Return} $Snap$ 
%\item[\textbf{End: extractView}]
\item[]

%\item[/* Lines 34 to 35 executed atomically */]
\item[\textbf{Procedure readTag()}:]
\STATE Send $(\quotes{readTag})$ to all 
\STATE \textbf{Wait until} receiving \\\hspace{0.8in} $\geq n-f$ $(\quotes{readAck}, *)$ msgs
\STATE \textbf{Return}  largest tag contained in $readAck$ msgs %\\
%\item[\textbf{End: readTag()}]

\item[] 
\item[]

%\item[/* Lines 37 \& 38 executed atomically */]
\item[\textbf{Procedure writeTag($tag$)}:]
\STATE Send $(\quotes{writeTag}, tag)$ to all 
%\STATE Wait for at least $n - f$ $(\quotes{wt\_ack}, tag)$ msgs
\STATE \textbf{Wait until} receiving \\\hspace{0.8in} $\geq n-f$ $(\quotes{writeAck}, tag)$ msgs
%\item[\textbf{End: writeTag}]

\item[] 

\item[/* Event handlers: executing in background */]
\item[/* All event handlers executed atomically */]
%\item[/* Executing in background */]
%\item[\textbf{Event Handlers:}] %/* Executing in background */ ]
\item[{\bf Upon receiving $(\quotes{value}, \tuple{u, ts})$} from $j$:]
\STATE  Add $\tuple{u, ts}$ into $V[j]$, $V[i]$  \\ 
\IF{$\tuple{u,ts}$ has \textit{not} been seen before}
\STATE Send $(\quotes{value}, \tuple{u,ts})$ to all
\ENDIF
\item[]

\item[{\bf Upon receiving $(``writeTag", tag)$} from $j$:]
\STATE $maxTag \gets \max\{maxTag, tag\}$
\STATE Send $(\quotes{echoTag}, tag)$ to all
\STATE Send $(\quotes{writeAck}, tag)$ to $j$ 
\item[]

\item[{\bf Upon receiving $(``echoTag", tag)$} from $j$:]
\STATE $maxTag \gets \max\{maxTag, tag\}$
% \STATE \lewis{Send $(\quotes{writeTag}, tag)$ to all}
% \STATE Send $(\quotes{writeAck}, tag)$ to $j$ 
\item[]

\item[{\bf Upon receiving $(\quotes{readTag})$ from $j$}:]
\STATE Send $(``readAck",  maxTag)$ to $j$
\item[]

\item[{\bf Upon receiving $(\quotes{goodLA}, r)$ from $j$}:]
\STATE $D[j] \gets V[j]^{\leq r}$ \COMMENT{borrow $j$'s view}
% Lewis: change completed to completeLA

~

\item[/* NOTE: All our event handlers are atomic;]
\item[]hence, when receiving a  $(\quotes{goodLA},*)$
\item[]Line 48 will be \underline{executed before} Line 29  
\item[] if there is a pending $LatticeRenewal()$ */
%Lewis: Don't remove these. I want to keep "return D[j]" at the left column
% \item[]
% \item[] 
\end{multicols}
\end{algorithmic}
\caption{ASO: Code for node $i$}
\label{algo:ASO}
\end{algorithm*}

\subsubsection{\alg{} Description and Pseudocode}

The pseudocode of \alg{} is presented in Algorithm \ref{algo:ASO}. We first describe key variables used, followed by the procedures and message handlers.

\vspace{3pt}
\noindent{\bf Variables}:~~
%\paragraph*{Variables}
% In our presentation, if a variable belongs to a node $i$, we use the subscript $i$ to denote it. If the identity is clear, then we may often omit the subscript for brevity. 
Each value (written by an \update{} operation) is associated with a timestamp of the form $\tuple{r, j}$, where $r$ is the tag and $j$ is the ID of the writer who initiates the \update{}. The exact value of the tag in the timestamp is determined in the \update{} operation.
For brevity, we often use value to denote a \textit{value-timestamp pair}.
For a set of values $H$, we use $H^{\leq r}$ to denote the set of values with tag at most $r$. 

At all time, each node $i$ keeps track of a vector $V_i$ of size $n$, which represents the vector of ``view'' at node $i$. 
% Roughly speaking, the view $V_i[j]$ represents $i$'s estimation on the values received at node $j$. 
Formally, for $j \in [n]$, $V_i[j]$ is the set of written and/or forwarded values that $i$ has received from node $j$. In our design, each node $i$ needs to forward a value it receives for the first time. In this case, we say a value is forwarded by $i$. 
Two other variables are related to $V_i$: (i) $V_i^{\leq r}$ is the vector of views with tag at most $r$, i.e., $V^{\leq r} = [V[1]^{\leq r}, V[2]^{\leq r}, \dots, \dots, V[n]^{\leq r}]$; and (ii) $D_i[j]$ is a \textit{particular view borrowed from node $j$} that can be ``safely'' returned. 
The meaning of ``safe'' will become clear when we discuss the lattice operation.

Each node also keeps track of a variable $maxTag$, which represents the largest tag it has ever received via $\quotes{writeTag}$ messages or $\quotes{echoTag}$ messages. Note that it is possible that there are some values with tag larger than $maxTag$ in $V_i$.
%In \alg, when a node completes a good lattice operation, it will send the view obtained to all other nodes via a $\quotes{completed}$ message. Each node stores a vector $D_p$ of size $n$. $D_p[j]$ is used to store the most recent view sent by node $j$ via a $\quotes{completed}$ message. 
%For example, $maxTag_i$ denotes the value of variable $maxTag$ at node $i$. 

%\paragraph*{Procedures}

\vspace{3pt}
\noindent{\bf Procedures}:~~
We explain two helper procedures, $Lattice(r)$ and $LatticeRenewal(r)$, and two interface procedures, \update$(v)$ and {\sc Scan}(). Other procedures are fairly straightforward from the pseudocode.
%A lattice operation is similar to the ELA algorithm for lattice agreement,  We say a lattice operation is good if it returns a non-empty view. Intuitively, a good lattice operation does not observe another lattice operation with higher tag. 

{\bf Lattice$(r)$}: 
Each node uses the $Lattice(r)$ procedure to run the $r$-th instance of lattice agreement, or in our terminology, lattice operation with tag $r$. 
The goal of a lattice operation is to solve lattice agreement, except that it is associated with an input tag $r$ and the termination condition depends on $r$ and the messages received, especially those with tag $\leq r$. 

Consider a Lattice$(r)$ invocation at node $i$. Node $i$ first writes the input tag $r$ to at least $n - f$ nodes. Then it waits until the \textit{equivalence quorum} predicate (Definition \ref{def:EQ}) becomes true for the first time. After that if the $maxTag$ value is strictly larger than $r$, then the lattice operation returns $\tuple{false, \emptyset}$. Otherwise, it returns $\tuple{true, V^*}$ where $V^* = V^{\leq r}$ is the vector that satisfies the equivalence quorum predicate. In this case, $Lattice(r)$ is said to be a good lattice operation, as defined next.
An important design choice here is that line 16 to 21 are executed \textit{atomically}. Therefore, once $V^*$ satisfies the equivalence quorum predicate for the first time, \underline{no other value is added to $V^*$}.

\begin{definition}
[Lattice Operation] We call each execution of the $Lattice(r)$ procedure as a \underline{lattice operation with tag $r$}. A lattice operation is \underline{good} if it returns $true$ at line \ref{line:lattice_return_true}.
\end{definition}

{\bf LatticeRenewal($r$)}: The $LatticeRenewal(r)$ procedure is also given a parameter $r$. 
It contains at most three lattice operations. If some lattice operation is good, it returns the view obtained by the good lattice operation, i.e., \textit{direct view}. If the first two lattice operation are \textit{not} good, then by definition, it means that node $i$ has observed a larger tag, i.e., condition at line \ref{line:good_lattice_condition} returns $false$. Therefore, node $i$ initiates the next lattice operation with tag equal to $maxTag$.
If the third lattice operation is also not good, then node $i$  waits for a $\quotes{goodLA}$ message from some other node $j$ to obtain a view from $j$'s good lattice operation. In this case, the view is called an \textit{indirect view} or a \textit{borrowed view}. 

It is fairly straightforward to see that our design satisfies the (P1) to (P3) stated earlier in Section \ref{sec:LR}. (P4) and the round complexity are less obvious, and depend on how we combine different components together.
Both \update{} and \scan operations use the $LatticeRenewal()$ procedure to obtain a view. This approach works owing to (P4).
Intuitively, the moment that \update{} and \scan obtains a view is the synchronization point.
%For \update{} operations, the view is not used, so the $LatticeRenewals()$ procedure acts as a synchronization point. 
%For \scan operations, it constructs the snapshot to be returned by extracting the most recent value for each node from the view returned. For an \update{} operation $op$, we call the invocation of the $Lattice()$ procedure at line \ref{line:initial_phase} as its initial phase. For both \update{} and \scan operations, we call the three iterations in the {\bf for} loop of the $LatticeRenewals()$ invocation as phase $1$, phase $2$ and phase $3$. 

{\sc \bf Update$(v)$}: To write value $v$,  node $i$ first obtains a tag by reading from at least $n - f$ nodes. Let $r$ denote the largest tag in the received $readAck$ messages. Then, $i$ constructs the timestamp of value $v$ as the $\tuple{r + 1, i}$ tuple. It sends value $v$ with its timestamp to all nodes. Then, a lattice operation with tag $r$ is invoked. This step is called the \textit{phase 0} lattice operation of the \update{} operation.
%%%% Change initial phase to phase 0 %%%%
After the phase 0 lattice operation, the \update{} obtains a new tag $r'$ and executes $LatticeRenewal(r')$.
The view returned by $LatticeRenewal(r')$ is \textit{not} used; and hence discarded.
Node $i$ returns the ACK to complete the \update{}.

In addition to the $LatticeRenewal()$ procedure, another subtle point is to execute the phase 0 lattice operation \textit{before} invoking $LatticeRenewal()$.
The way we devise them ensures that
for each tag,
there is a \textit{good} lattice operation. Recall that a lattice operation is good if it returns $true$ at line 17.
Lemma \ref{lem:nonskipping_good_execution} presented later explains this statement in more details.

{\sc \bf Scan()}: The code for a \scan operation is quite simple. It first obtains a tag $r$ similarly.
Then, it executes the $LatticeRenewal(r)$. After $LatticeRenewal(r)$ returns a view $scanView$, the node takes the most recent value by each node in $scanView$ by executing the $extract(scanView)$ procedure.

%\paragraph*{Message Handlers} 

% \alg{} has four message handlers that implements the logic upon receiving messages of different types. 

\vspace{3pt}
\noindent{\bf Message Handlers}:~~
All the handlers execute in the background; hence, even if a node does not have a pending \update{} or \scan operation, it still processes messages. Moreover, all the handlers are executed \textit{atomically}, i.e., during the period that a handler is executing, no other part of the code can take step. All the handlers should be clear from the code. One subtle part to note is that a node does \textit{not} update its $maxTag$ variable when it receives a value with a larger tag from a $\quotes{value}$ message. The $maxTag$ variable is only updated when a node receives a $\quotes{writeTag}$ message or $\quotes{echoTag}$ message. This design helps \alg{} achieve the desired round complexity. Especially, we rely on it to prove Lemma \ref{lem:nonskipping_good_execution}. For completeness, we also present detailed description of message handlers is in Appendix \ref{app:message_handlers}.

\subsection{Proof of Correctness}

For correctness, we need to  prove termination and construct a linearizable sequence of \update{} and \scan operations for any execution. We first discuss important definitions and properties of our algorithm to facilitate the proof. Then, we prove termination and show the linearization construction.

\vspace{3pt}
\noindent\textbf{Useful Definitions: Tags and Views}:~~
We say an \update{} or \scan is direct if its $LatticeRenewal()$ procedure returns at line \ref{line:direct}; otherwise, it is indirect. Intuitively, an operation is direct if it there is a good lattice operation during $LatticeRenewal()$. 
% Recall that $LatticeRenewal()$ contains several phases with different tags. 
We define the tag of an operation and value as following.

\begin{definition} 
[Tag of \update{} or \scan] \label{def:tag_operation} The tag of an \update{} or \scan operation is the tag of its last lattice operation.
\end{definition}

\begin{definition}
[Timestamp/Tag of a value] The timestamp of a value is the $\tuple{r+1, i}$ (tag-ID tuple) at line \ref{line:value_ts} in the \update$(v)$ procedure. The tag of a value is defined as the tag contained in its timestamp. \underline{For value $v$, we use $ts_v$ to denote its timestamp.}
\end{definition}

Now we introduce an important concept, \textit{view}, that is used throughout our proof. 
% Generally speaking, a view is defined as the set of all past values observed by some node or an operation.

\begin{definition}
[View] 
\label{def:view}
We define the views for a node and operations as below:
\vspace{-0.1in}
\begin{itemize}
    \item For a node $i$, its view is defined as the set $V_i[i]$.
    \item For a good lattice operation with tag $T$ ($Lattice(T)$) at node $i$, its view is defined as the set of values with tag at most $T$ in $V_i[i]$ right after completing line \ref{line:lattice_termination_condition}, i.e., $V_i[i]^{\leq T}$.
    \item For an \update{} or \scan operation, its view is defined as the set returned by its $LatticeRenewal()$ procedure.
\end{itemize}
%(1) For a node $p$, its view is defined as the set $V_p[p]$. (2) For a good lattice operation with tag $T$, i.e., $Lattice(T)$, by node $p$, its view is defined as the set of values with tag at most $T$ in $V[p]$ after completing line \ref{line:lattice_termination_condition}, i.e., $V_p[p]^{\leq T}$. (3) For an \update{} or \scan operation $op$ by node $p$ with tag $T$, its view is defined as the set returned by its $LatticeRenewals()$ invocation. %(3) For an \update{} or \scan operation $op$ by node $p$ with tag $T$, if it is direct, its view is defined as the view of its last lattice operation. Otherwise, its history is defined as the set $D[q]$ when it receives $(\quotes{completed}, T)$ by some node $q$ at line \ref{line:wait_borrow}.
\end{definition}

% \begin{lemma}
% The view of an \update{} or \scan operation must be the view of some good lattice operation.
% \end{lemma}

% \begin{observation}
% The view of an \update or \scan operation must be the view returned by some good lattice operation. 
% \end{observation}

% \begin{observation}
% If an \update or \scan is direct, then its last lattice operation must be a good lattice operation. 
% \end{observation}

% \xiong{We assume each value is unique}

% For two timestamps $ts_1 = ~<t_1, p_1>$ and $ts_2 = ~<t_2, p_2>$, we say $ts_1 < ts_2$ iff $t_1 < t_2 \vee (t_1 = t_2 \wedge p_1 < p_2)$. %For any two values $u$ and $v$ written by a same node, we say $u \leq v$ if $ts_u \leq ts_v$. \lewis{Need to explain the ordering of timestamps.}

% For any two views $H_1$ and $H_2$, we say that $H_1 \leq H_2$ if $H_1[i] \leq H_2[i]$ for each $i \in [n]$. We say $H_i$ and $H_j$ are comparable if either $H_1 \leq H_j$ or $H_j \leq H_1$; otherwise, incomparable.  

% \lewis{In Def. 14, do we want to use the term "witness set"? Or just equivalence quorum?} \xiong{Equivalence quorum is better. Should we change the def 15 based on the equivalence predicate?}
% \lewis{That's a good idea, which will make the presentation more coherent.}

% \begin{definition}
% [Equivalence quorum of a good lattice operation] For a good lattice operation by node $p$, its equivalence quorum is defined as the set $Q$ of at least $n - f$ nodes  (including $p$) such that $V_p[p] = V_p[q]$ for each $q \in Q$ when line \ref{line:lattice_termination_condition} completes. 
% \end{definition}

We present two lemmas on properties of the tags. Lemma \ref{lem:value_tag_nonskipping} follows directly from line \ref{line:tag_increment}.

\begin{lemma} \label{lem:value_tag_nonskipping}
The tags of values are non-skipping, i.e., if there is a value with tag $T \geq 1$, then there is also a value with tag $T - 1$. 
\end{lemma}

%\begin{proof}
%By line \ref{line:tag_increment}, the tag of values increases by at most one when an \update{} operation is invoked. 
%\end{proof}

The proof of the following lemma explains why we need the phase 0 lattice operation.

\begin{restatable}{lemma}{consecutiveGoodLA} \label{lem:nonskipping_good_execution}
If the largest tag in the system is $T$ at time $t$, i.e., $\max_{i \in [n]} maxTag_i = T$ at time $t$, then for each $1 \leq z \leq T - 1$, there is a good lattice operation with tag $z$ before time $t$.
\end{restatable}

\begin{proof}
To prove the lemma, we first prove the following claim. 

\begin{claim}
\label{claim:tag_before_t}
If the largest tag in the system is $T$ at time $t$, then there exists a good lattice operation with tag $T - 1$ that completes before $t$. 
\end{claim}

\begin{proof}[Proof of Claim \ref{claim:tag_before_t}]
Observe that a value with tag $T$ must be sent by some node, since we consider only crash failures. 
Let node $i$ be the first node that sends tag $T$ to all other nodes inside the $writeTag$ procedure in operation $Op$.
Since $Op$ is the first operation to send tag $T$, 
$Op$ must be an \update{}  operation. Let $L_0$ denote the phase 0 lattice operation inside $Op$. 
When $Op$ executes line \ref{line:tag_increment}, we have $maxTag_i < T$, %\sapta{**This might not be true if some other node successfully completes writetag before $p$. We might need to do the induction on the first node that successfully completes writetag of $T$**}
and $r_i = T$ when line \ref{line:tag_increment} completes. This implies that $r_i = T-1$ at line \ref{line:update_readTag}.  Thus, the phase 0 lattice operation $L_{0}$ of $Op$ at line \ref{line:initial_phase} must have tag $T - 1$.
That is, node $i$ invokes $Lattice(T-1)$ as the phase 0 lattice operation.

Since by assumption, $Op$ is the first operation that proposes tag $T$ and it proposes tag $T$ after line \ref{line:tag_increment} (through the $writeTag$ step at line \ref{line:lattice_writeTag}), no node has proposed tag $T$ before the execution of line \ref{line:tag_increment}.
Hence, during the execution of the phase 0 lattice operation $L_0$ at line \ref{line:initial_phase}, we have $maxTag_i \leq T - 1$ . Recall that the tag of $L_0$ is $T - 1$.  %i.e., the local variable $R$ of $L_0$ is $T - 1$. 
This implies that at line \ref{line:good_lattice_condition} inside $L_0$, we have $maxTag_i \leq T-1 = r_i$. 
%\lewis{We need more explanation here, i.e., explain the relationship between $maxTag$ and $R$}
Thus, $L_0$ is a good lattice operation. Moreover, by assumption, $L_0$ completes before time $t$. This proves Claim \ref{claim:tag_before_t}.
\end{proof}

Since tags are non-skipping by Lemma \ref{lem:value_tag_nonskipping}, applying Claim  \ref{claim:tag_before_t} inductively gives us that for each $1 \leq z \leq T - 1$, there exists a good lattice operation with tag $r$ before time $z$.
\end{proof}

% \lewis{Is it more intuitive to express this way: the first operation proposing $T$ must contain a good lattice operation for $T-1$?} \xiong{We can also do this, then we need the above lemma as a corollary since we will directly use it in Lemma \ref{lem:operation_termination}}

% \subsubsection{Termination}
% \label{s:aso-termination}

% We are now ready to prove the termination property. That is, we prove that each operation eventually terminates. We will analyze the round complexity later in Section \ref{s:aso-round}. In the presentation below,  for both \update{} and \scan operations, we call the three invocation of the lattice operation (in the {\bf for} loop of the $LatticeRenewal()$) as phase 1, phase 2 and phase 3 lattice operations.

% In the following proofs, we rely on  the assumption of reliable FIFO channels. Particularly, once $L_0$ completes at node $i$, eventually all nonfaulty nodes will receive its $\quotes{goodLA}$ message for $L_0$ and all the messages sent by node $i$ before the $\quotes{goodLA}$ message. 

\vspace{3pt}
\noindent\textbf{Termination}:~~We show that each operation eventually terminates if each lattice operation terminates. We will show that each lattice operation takes $O(\sqrt{k})$ rounds later.

\begin{restatable}{lemma}{terminationASOLemma}[Termination] \label{lem:operation_termination}
If each lattice operation eventually terminates, then \update{} and \scan operations in Algorithm~\ref{algo:ASO} eventually terminate.
\end{restatable}

\begin{proof}[Proof Sketch]
We show that each $LatticeRenewal()$ eventually terminates. The only blocking part is line \ref{line:wait_borrow}. Lemma \ref{lem:nonskipping_good_execution} implies that if a node observes a tag $T$, then it must be able to borrow a good view for some tag smaller than $T$. 
\end{proof}

\vspace{3pt}
\noindent\textbf{Useful Lemmas: Comparable Views}:~~
Next we prove Invariant \ref{invariant}, which is formally stated in Lemma \ref{lem:lattice_comparable_view}.

\begin{restatable}{lemma}{viewFromGoodLattice} \label{lem:view_from_good_lattice}
The view of each \update{} or \scan operation is the same as the view of some good lattice operation. %Thus, the view of each \update or \scan operation is the view returned by some good lattice operation.
\end{restatable}

Similar to Lemma \ref{lem:comparable_views_for_same_node}, we obtain the following lemma on the views (bounded by tag $T$) at node $i$ and $j$ due to our assumption on FIFO channel. Its proof is in Appendix \ref{app:comparable_tagged_views_for_same_node}.%\xiong{$V_p[s]^{\leq T}$ and $V_q[s]^{\leq T'}$ are not necessarily comparable, which is why our previous ASO algorithm does not work}

\begin{restatable}{lemma}{comparableTagViewsSameNode} \label{lem:comparable_tagged_views_for_same_node}
For any two nodes $i$ and $j$ and tag $T$, fix any time $t$ and $t'$, the set $V_i[s]^{\leq T}$ at time $t$ and the set $V_j[s]^{\leq T}$ at time $t'$ are comparable for each node $s$.
%\lewis{It seems to me that this lemma holds even if $s$ is faulty?}
\end{restatable}

% \begin{proof}

% The value of set $V_i[s]^{\leq T}$ is modified only when $i$ receives a new value with tag $\leq T$ from node $s$. Since the communication is FIFO, the set $V_i[s]^{\leq T}$ at time $t$ must be the same as the set $V_s[s]^{\leq T}$ at some time $t_i$ and the set $V_j[s]^{\leq T}$ at time $t'$ must be the same as the set $V_s[s]^{\leq T}$ at some time $t_j$. The set $V_s[s]^{\leq T}$ is non-decreasing. Thus, $V_i[s]^{\leq T}$ at time $t$ must be comparable with $V_j[s]^{\leq T}$ at time $t'$. 

% % For the sake of contradiction, suppose there exists a value $v \in V_i[s]^{\leq T} - V_j[s]^{\leq T}$ and a value $u \in V_j[s]^{\leq T} - V_i[s]^{\leq T}$. For node $s$, it either sends value $v$ before $u$ or $u$ before $v$. Suppose node $s$ sends $v$ before $u$, w.l.o.g. Because the communication channel is FIFO, node $j$ must receive value $v$ before receiving $u$. This contradicts the assumptions that (i) there exists a value $v \in V_i[s]^{\leq T} - V_j[s]^{\leq T}$; and (ii) $u \in V_j[s]^{\leq T}$. Thus, $V_i[s]^{\leq T}$ and $V_j[s]^{\leq T}$ must be comparable. 
% \end{proof}

Next, we prove an important lemma which shows our key usage of lattice operation and the equivalence quorum technique. Lemma \ref{lem:lattice_comparable_view} is a formal statement of Invariant \ref{invariant}. We put its full proof in Appendix \ref{app:lattice_comparable_view}.

\begin{restatable}{lemma}{laComparableViews} \label{lem:lattice_comparable_view}
The views returned by all good lattice operations are comparable. 
\end{restatable}
\begin{proof}[Proof Sketch]
For any two lattice operations $Op_i$ and $Op_j$ with tag $T_i$ and $T_j$, if $T_i = T_j$, then Lemma \ref{lem:comparable_tagged_views_for_same_node} and the equivalence quorum predicate imply that their views must be comparable. Otherwise, assume w.l.o.g $T_i < T_j$. Our algorithm guarantees that the view of $Op_i$ must be a subset of the view of $Op_j$. Intuitively, the fact that $Op_i$ does not observe $T_j$ at line \ref{line:good_lattice_condition} implies that $Op_j$ must complete its line \ref{line:lattice_writeTag} step after $Op_i$ has completed. This ensures that $Op_j$ must have received all values in the view of $Op_i$ when $Op_j$ starts line \ref{line:lattice_termination_condition}. 
\end{proof}

Lemma \ref{lem:comparable_views} immediately follows from Lemma \ref{lem:view_from_good_lattice}  and \ref{lem:lattice_comparable_view}. Lemma \ref{lem:comparable_views} allows us to construct a linearization of \scan and \update{} operations later. 
% Intuitively, this is because the view of each operation can be deemed as the synchronization point.

\begin{lemma} \label{lem:comparable_views}
The views returned by all \update{} and \scan operations are comparable. 
\end{lemma}
%\begin{proof}
%Immediate from Lemma \ref{lem:view_from_good_lattice}  and \ref{lem:lattice_comparable_view}.
%\end{proof}

% \begin{lemma} \label{lem:tag_order}
% For any two operations $op_i$ with tag $tag_i$ and $op_j$ with tag $tag_j$, respectively. If $op_i \rightarrow op_j$, then $tag_i \leq tag_j$. 
% \begin{proof}
% We first observe that when an operation completes, its tag must be in the history of at least $n - f$ nodes. Thus, after $op_i$ completes, when $op_j$ executes $getTag()$, it must be able to get a value $\geq tag_i$ due to quorum intersection. 
% \end{proof}
% \end{lemma}

% \lewis{Maybe we want to find another variable for $<v, ts>$ pair? Otherwise, $v$ would be confused with value. }

% \xiong{Since each value is associated with a unique timestamp, I think we can simply use value to mean the value-timestamp pair. }

% \lewis{It won't be precise that way then. For example, if you read Lemma 17. The first thing readers know about $v$ and $v'$ is some value that is $update$'s input. It doesn't have the notion of timestamp there. It doesn't affect correctness, but distinguishing value and value-ts pair would help reader digest the proof.}

%\subsubsection{Useful Lemma: Visible Views}

\vspace{3pt}
\noindent\textbf{Useful Lemma: Visible Views}:~~
To respect the atomicity semantics, we also need to ensure that (i) once an \update{} is completed, then its value is visible to subsequent \scan's; and (ii) once a \scan reads certain set of values, these values are also visible to subsequent \scan's. We prove these two through the usage of views and tags. 

The lemma below is straightforward from the code. Refer Definition \ref{def:view} for view definition. 

% Recall that view is defined in Definition \ref{def:view}

% , the view of a good lattice operation with tag $T$ is the set of values with tag at most $T$ in $V_i[i]$ right after completing line \ref{line:lattice_termination_condition}.

\begin{lemma} \label{lem:good_lattice_downward_validity}
For a good lattice operation $Op$ by node $i$ with tag $T$, let $H$ denote the view of node $i$ right before $Op$ execute  line \ref{line:lattice_termination_condition} and $H_{Op}$ denote the view of $Op$. Then, $H^{\leq T} \subseteq H_{Op}$. 
\end{lemma} 

The next lemma is the main reason that we need to have three lattice operations in the $LatticeRenewal()$ procedure.

\begin{restatable}{lemma}{outputDomiInput}
\label{lem:output_view_dominate_input_view}
Let $Op$ be an \update{} or \scan operation by node $i$ with tag $T$. Let $H_{Op}$ denote its view. Let $H$ denote the view of node $i$ right before $Op$ executes its $LatticeRenewal()$ invocation. Then, $H^{\leq T} \subseteq H_{Op}$. 
\end{restatable}

\begin{proof}
We assume that $Op$ is an \update{} operation. The proof for the other case is similar. 
If $Op$ is an \update{} with a direct view, the claim follows from Lemma \ref{lem:good_lattice_downward_validity}, since by definition of a direct view, $Op$ obtains a view from its good lattice operation.  
Now, consider the case when $Op$ completes with an indirect view. By construction,
$Op$ must continue to phase 3. Let $R_1$, $R_2$ and $R_3$ denote the tags for each of the three lattice operations in $Op$'s invocation of $LatticeRenewal()$, respectively. Then, by Definition \ref{def:tag_operation}, $T = R_3$, the tag of the last lattice operation in $Op$.
Moreover, $Op$ must have received $(\quotes{goodLA}, R_3)$ message %at line 14 
from some other node $j$. Let $L_j$ denote this particular lattice operation by node $j$. By construction, $L_j$ is a good lattice operation with tag $R_3$.  Then, we prove the following claim.

\begin{claim}
\label{claim:noR3}
Tag $R_3$ was \textit{not} known by node $i$ during its phase 1 lattice operation.
\end{claim}

\begin{proof}[Proof of Claim \ref{claim:noR3}]
First observe that $R_1 < R_2 < R_3$, since none of the lattice operation in $i$'s $LatticeRenewal()$ is good.
The fact that $Op$ obtains tag $R_2$ such that $R_2 < R_3$ during phase 2 implies that when $Op$ completes its phase 1, it has \textit{not} learned $R_3$; otherwise, it would not proceed to phase 2 with tag $R_2$, since $R_2 < R_3$.
\end{proof}

Let $L_i$ denote the lattice operation by $Op$ in phase 1 at node $i$. Consider the $writeTag$ procedure in $L_i$ and $L_j$. Let $Q_i$ and $Q_j$ denote the set of nodes that sent the $\quotes{writeAck}$ messages in responding to the $\quotes{writeTag}$ message of $L_i$ and $L_j$, respectively. Since both set of nodes are of size at least $n-f$, there exists a nonfaulty node $s \in Q_i \cap Q_j$ such that node $s$ must have received the $\quotes{writeTag}$ message from $L_j$ after sending $\quotes{writeAck}$ message in responding to the $\quotes{writeTag}$ message of $L_i$. Otherwise, $Op$ would obtain tag $R_3$ for phase 2, a contradiction to Claim \ref{claim:noR3}. %\lewis{Xiong: I think this part needs our new edit on Friday, i.e., Line 44 in the algorithm. Please check.} \xiong{This part is not affected. } \lewis{I meant without Line 44, this proof is incorrect. So with Line 44, we should be good.}

%\lewis{or some tag $\leq R_3$? I'm assuming that you meant in this case, tag of $op'$ is $R_3$? By it could be something larger?}

Since communication is reliable and FIFO, and node $i$ sends all values in $H$ before sending the $\quotes{writeTag}$ message in lattice operation $L_i$, node $s$ must receive all values in $H$ and sends out to all other nodes before sending the $\quotes{writeAck}$ message in responding to the $\quotes{writeTag}$ message of $L_j$. Thus, node $j$ must have received all values in $H$ before it completes line \ref{line:lattice_writeTag} of $L_j$.
%executes line 29 of lattice operation $l'$.
Let $H_j$ denote the view of lattice operation $L_j$. Since $L_j$ is a good lattice operation, Lemma \ref{lem:good_lattice_downward_validity} implies that $H^{\leq R_3} \subseteq H_j$. 

%Now we show that $H_j \subseteq H_{Op}$. 
By assumption, node $i$ borrows $j$'s view after it has received the $(\quotes{goodLA}, R3)$ at line \ref{line:wait_borrow}. Thus, $H_j \subseteq H_{Op}$. Therefore, $H^{\leq R_3} = H^{\leq T} \subseteq H_j \subseteq H_{Op}$. 
% \lewis{Xiong: please check this paragraph. I changed it a bit. You assumed that $i$ may borrow a view from some node other than $k$; however, didn't we assume $i$ borrow $j$'s view in the beginning of this proof? Feel free to change it back if I miss something.}
%The history of $Op$ is the same as the $D[q]$ at line 15 after it has received the $(\quotes{goodLA}, R3)$ at line 14. By FIFO communication, the lattice operation that sends $D[q]$ must have tag at least $R_3$. Thus, $H' \subseteq H_{op}$. Therefore, $H^{\leq R_3} = H^{\leq T} \subseteq H_{op}$. 
\end{proof}

% The following Lemma directly follows form Lemma \ref{lem:output_view_dominate_input_view}.

The above lemma immediately implies that the value of an \update{} operation must belong to the view obtained by the \update{} operation. 

% \begin{restatable}{lemma}{downwardValidity} \label{lem:downward_validity}
% Let $Op$ be an \update{} operation by node $i$ that writes value $v$ with timestamp $ts$ and has view $H_i$. Then, $\tuple{v, ts} ~ \in H_i$.
% \end{restatable}

%\begin{proof}
%Let $ts$ denote the timestamp of $v$.  Lemma \ref{lem:output_view_dominate_input_view} implies that $<v, ts>$ must be in the history of $op_i$. Thus, $v ~\leq view_i[i]$.
%\end{proof}

%The following Lemma follows from the assumption on FIFO channel. \lewis{Do we use this lemma anywhere?}

% \begin{lemma} \label{lem:local_subset_remote}
% For each node $i$ and $i$, $V_i[j] \subseteq V_j[j]$ at any time.
% \end{lemma}

Following the convention, we say that $Op \rightarrow Op'$ iff the response (or completion) time of $Op$ occurs before the invocation time of $Op'$. The following Lemma immediately follows from the usage of $writeTag$ and $readTag$ procedures.

\begin{lemma} \label{lem:tag_order}
For any two operations $Op_i$ with tag $T_i$ and $Op_j$ with tag $T_j$, respectively. If $Op_i \rightarrow Op_j$, then $T_i \leq T_j$. 
\end{lemma}

%\begin{proof}
%First observe that when an operation completes, its tag must be learnt by at least $n - f$ nodes through its $writeTag$ procedure. Thus,  when $Op_j$ executes $readTag$ procedure, it must be able to get a $T_j \geq T_i$ due to quorum intersection. 
%\end{proof}

Now we prove the important lemma on views being ``visible'' to subsequent operations.

\begin{restatable}{lemma}{viewOrder} \label{lem:view_order}
For any two operations $Op_i$ of node $i$ and $Op_j$ of node $j$ with views $H_i$ and $H_j$, respectively. If $Op_i \rightarrow Op_j$, then $H_i \subseteq H_j$. 
\end{restatable}
\begin{proof}
Let $T_i$ and $T_j$ denote the tag of $Op_i$ and $Op_j$, respectively. 
Consider the following two cases. 
Let $H$ denote the view of node $j$ right before $Op_j$ invokes $LatticeRenewal()$. Lemma \ref{lem:output_view_dominate_input_view} implies that $H^{\leq T_j} \subseteq H_j$. To prove the lemma, we need to show that $H_i \subseteq H^{\leq T_j}$. 

\begin{itemize}
    \item \textit{Case 1}: $Op_i$ obtains a direct view. 
    
    Consider $Op_i$'s last lattice operation $L_i$. Then $L_i$ is a good lattice operation with tag $T_i$. By definition, $H_i$ is also the view of $L_i$. 
    Let $Q_i$ denote the \textit{equivalence quorum} for $L_i$. Then, we have $H_i = V_i[w]^{\leq T_i} \subseteq V_w[w]^{\leq T_i}$ for each node $w \in Q_i$ when $L_i$ completes. Let $Q_j$ denote the set of nodes which send $\quotes{readAck}$ for the first $\quotes{readTag}$ message of $Op_j$. Since $Op_j$ starts after $Op_i$ completes, there exists a nonfaulty node $s \in Q_i \cap Q_j$ such that node $s$ sends out all the values in its current view, which must include all the values in $H_i$, to all the other nodes before sending the $\quotes{readAck}$ message for $Op_j$'s first $\quotes{readTag}$ message. By FIFO channels, node $j$ must have received all the values in $H_i$ before it completes the $writeTag$ procedure of $Op_j$. Lemma \ref{lem:tag_order} implies that $T_j \geq T_i$. This together with the observation that the largest tag in $H_i$ is $T_i$,  we have $H_i \subseteq H^{\leq T_j}$. 
    %\lewis{Why do we use Lemma \ref{lem:tag_order} It doesn't seem relevant here?}

    \item \textit{Case 2}: $Op_i$ obtains an indirect view. 
    
    The view $H_i$ of $Op_i$ is the same as the view of some good lattice operation $L$ and $Op_i$ must have received the $\quotes{goodLA}$ message sent by $L$. Thus, when $Op_i$ completes, $L$ have completed its execution of line \ref{line:broadcast_decision}. Then, by a similar argument in case 1, $H_i \subseteq H^{\leq T_j}$.
\end{itemize}
\end{proof}

\vspace{3pt}
\noindent\textbf{Construction of a linearization}:~~
\label{s:aso-linearization}
For a given execution, we construct a sequence $\sigma$ of all \update{} and \scan operations in the execution such that $\sigma$ preserves the semantics of atomic snapshot object. The construction is similar to one from \cite{attiya1995atomic}, and presented below.

\begin{itemize}
    \item \textbf{Insert \scan operations}: First, we construct a sequence $\sigma'$ which includes all \scan operations. The \scan operations are ordered in $\sigma'$ according to the order of their associated views. Specifically, for any two \scan operations $Sc_i$ and $Sc_j$ that have view $H_i$ and $H_j$, respectively, if $H_i < H_j$, then $Sc_i$ appears before $Sc_j$ in $\sigma'$. If $H_i = H_j$ and $Sc_i \rightarrow Sc_j$, then $Sc_i$ appears before $Sc_j$ in $\sigma'$. Otherwise, $Sc_i$ and $Sc_j$ are ordered arbitrarily. 
    
    \item \textbf{Insert \update{} operations}:
    Second,
    we insert all \update{} operations into $\sigma'$. Consider an \update{} operation $Op$ that writes $v$ with timestamp $ts_v$. We insert $Op$ after all \scan operations whose view do not include $\tuple{v, ts_v}$ and before all \scan operations whose view contains $\tuple{v, ts}$. That is, $Op$ is inserted just before the first \scan operation in $\sigma'$ such that its view contains $\tuple{v, ts}$. For any two \update{} operations $Op_1$ and $Op_2$ that fit between the same pair of \scan operations. If $Op_1 \rightarrow Op_2$, then we insert $Op_1$ before $Op_2$ in sequence $\sigma$. Otherwise, $Op_1$ and $Op_2$ are ordered arbitrarily. 
\end{itemize}

Similar to the proof given in \cite{attiya1995atomic}, the proof of the following theorem uses Lemma \ref{lem:comparable_views}.  \ref{lem:output_view_dominate_input_view} and \ref{lem:view_order} to show that $\sigma$ is a linearizable sequence. We put it in Appendix \ref{app:mainTheoremASO}

\begin{restatable}{theorem}{mainTheoremASO} \label{theo:mainTheoremASO}
\alg{} (Algorithm~\ref{algo:ASO}) implements an atomic snapshot object.
\end{restatable}
% \begin{proof}
% The proof is direct from Lemmas~\ref{lem: sequential specification proof} and~\ref{lem: real time ordering}
% \end{proof}

\subsection{Round Complexity} 
\label{s:aso-round}

% \begin{observation}
% When a node $p$ obtains a tag $t$ from node $q$ in $getTag()$, the history of $q$ before $q$ sending out tag $t$ to node $p$ must be received by all correct nodes within time $D$.  
% \end
Now we analyze the round complexity of our algorithm. We assume the local computation time is negligible compared with the message delay. 
We show that each lattice operation takes $O(\sqrt{k})$ rounds. The proofs of the two lemmas below are in Appendix \ref{app:later_operation_bigger_tag} and \ref{app:completed_update_known_to_all_after_D}.

\begin{restatable}{lemma}{laterOperationBiggerTag}  \label{lem:later_operation_bigger_tag}
Suppose there exists a lattice operation that starts at time $t$ with tag $T$, then any \update{} operation starting after time $t + D$ must assign a tag $> T$ for its value.  Thus, all values with tags at most $T$ must have been sent out by time $t + D$. 
\end{restatable}
% \begin{proof}
% Let operation $op$ denote such a lattice operation. Since $op$ sends its tag $T$ to all in the $writeTag$ function at line 28, by time $t + D$, each correct node must have received tag $T$. Thus, any \update{} operation that starts after time $t + D$ must obtain a tag greater than $T$ for its value. Thus, all values with tags at most $T$ must come from \update{} operations that start before time $t + D$. Since local computation does not take time, all such values must be sent out at line \ref{line:brodcast_val} before time $t + D$. 
% \end{proof}

\begin{restatable}{lemma}{knownToAllAfterD} \label{lem:completed_update_known_to_all_after_D}
Let $Op$ denote \update$(v)$ operation. If $Op$ completes before time $t$, then for each nonfaulty node $i$, $\tuple{v, ts_v} ~\in V_i[j]$ for each nonfaulty node $j$ by the end of time $t + 2D$. 
\end{restatable}
% \begin{proof}
% Since $op$ completes before time $t$, $<v, ts_v>$ must be sent to node $q$ before time $t$ and must be received by node $q$ by time $t + D$. Thus, node $p$ must receive $<v, ts_v>$ from node $q$ by the end of $t+ 2D$. 
% \end{proof}

Recall that the exposed value is introduced in Definition \ref{def:new_value}. 

\begin{restatable}{lemma}{roundComplexityASO} \label{lem:worst_round_complexity_multishot}
Each lattice operation takes $O(\sqrt{k})$ message delays in the worst case. 
\end{restatable}
\begin{proof}
Let $L_i$ be a lattice operation at node $i$. 
Suppose $L_i$ starts at time $t$ with tag $T$. According to the condition at line \ref{line:lattice_termination_condition}, the termination of $L_i$ only depends on values with tags at most $T$.  Thus, we do not need to consider the values from \update{} operations that start after time $t + D$ by Lemma \ref{lem:later_operation_bigger_tag}. That is, for the termination of $L_i$, we only need to consider \update{} operations that start before time $t + D$.  Now, we prove an important claim.

\begin{claim}
\label{claim:atmostk}
There are at most $k$ exposed values with tag $\le T$ in intervals after time $t + 4D$.
\begin{proof}[Proof of Claim \ref{claim:atmostk}]
%Since the values from \update{} operations that start after time $t + D$ must have tag greater than $T$, we only need to consider values from \update{} operations that start before time $t + D$. 
By Lemma \ref{lem:completed_update_known_to_all_after_D}, all values with tags at most $T$ from \update{} operations that have completed before $t + D$ must be contained in $V_i[j]$ for each pair of nonfaulty nodes $i$ and $j$ by time $t + 3D$, i.e., known by all nonfaulty nodes. Thus, by definition of the exposed values, we have that values from \update{} operations that have completed before time $t + D$ \textit{cannot} be exposed values in intervals after time $t + 3D$. Since the values from \update{} operations that start after time $t + D$ must have tag greater than $T$, only values from \update{} operations that \underline{start before time $t + D$ and have not completed by time $t + D$} can be exposed values in intervals after time $t + 3D$. 
Let $U$ denote the set of these values.
Moreover, by definition, there can at most one such \update{} operation per node. 

%with tag at most $T$ from \update{} operations that start before time $t + D$ and have not competed by time $t + D$. 
Consider an arbitrary value $v \in U$ and the \update$(v)$ operation. We show that if \update$(v)$ is from a nonfaulty node, then value $v$ cannot be an exposed value for intervals after time $t + 3D$. Consider lines \ref{line:update_readTag} to \ref{line:brodcast_val} of \update$(v)$, since local computation takes negligible time, and line \ref{line:update_readTag} takes at most $2D$ time, value $v$ must be sent to all other nodes at line \ref{line:brodcast_val} before time $t + 3D$. Thus, by time $t + 4D$, value $v$ must be known by all nodes that have not crashed at this time, including all the nonfaulty nodes. Therefore, $v$ cannot be an exposed value in intervals after time $t + 4D$. 
Thus, a value in $v \in U$ can be an exposed value for intervals after time $t + 4D$ iff \update(v) is from a faulty node. This proves the claim. %Therefore, there are at most $k$ new values with tag at most $T$ for intervals after time $t + 4D$. 
\end{proof}
\end{claim}

The above Claim implies that after time $t + 4D$, if we only consider values with tag at most $T$, Lemma \ref{lem:termination_without_newValue}, \ref{lem:failureChainLength} and \ref{lem:unique_failure_chain} still hold. Thus, similar to the proof in Lemma \ref{lem:round_complexity_single_shot}, $l_p$ must terminate in $O(\sqrt{k})$ rounds.
\end{proof}

%\lewis{Consider moving this paragraph outside the proof. People will usually skip the proof it it make sense. We should have important argument outside of the proof.} \xiong{Agree. }
The proof of Claim \ref{claim:atmostk} also explains why we need to put line \ref{line:brodcast_val} before the initial lattice operation at line \ref{line:initial_phase}. If we switch the order of line \ref{line:brodcast_val} and line \ref{line:initial_phase}, then we cannot guarantee that value $v$ is sent to all the other nodes before time $t + 3D$, even though \update(v) starts before time $t + D$.

% \begin{lemma} \label{lem:amortized_round_complexity_multishot}
% Suppose that each node has $\Omega(\sqrt{k})$ operations. Each operation has amortized round complexity. 
% \end{lemma}

% \begin{theorem}
% There exists an implementation of atomic snapshot objects in message-passing systems such that each \update{} or \scan operation takes $O(\sqrt{k})$ rounds, where $k$ is the number of actual faults in the system. Thus, each \update{} or \scan takes $O(1)$ rounds with no crash failure and takes amortized $O(1)$ rounds if each node has $\Omega(\sqrt{k})$ operations. 
% \end{theorem}

% \lewis{Need to mention no-failure round complexity.}

% \input{ByzantineLattice}

% \subsection{Implementation of the assert transfer object for cryptocurrency}
%\input{related}

% \section{Algorithms}
% \section{Proof of Correctness}
% \input{related}

%\section{Conclusion and Future Work}

%1) One interesting problem left is to implement Byzantine tolerant atomic snapshot object in message passing systems. 

%2) Lower bound for lattice agreement problem. Is the $O(\sqrt{k})$ upper bound optimal for early stopping algorithms? 

%%
%% Bibliography
%%

%% Please use bibtex, 

\bibliography{references} 

\appendix
\newpage

% \section{Applications}

% \section{Generalized Lattice Agreement}

% There is a transformation between atomic snapshot object implementation and generalized lattice agreement. 

\section{Linearizable Update-Query State Machines} \label{app:linearizable_UQ}
In this section, we show how to implement a linearizable update-query state machine using our atomic snapshot algorithm. A update-query state machine only supports two types of operations: update and query. It does not support update and query mixed operations. It also assume that all updates are commutable, so the order of updates does not matter. Many data structures such as sets,
sequences, certain types of key-value tables, and graphs \cite{shapiro2011convergent}
can be designed with commuting updates. 

The implementation, shown in Algorithm \ref{algo:linearizable_UQ_state_machine}, is almost same as the atomic snapshot implementation in Algorithm \ref{algo:ASO}, except that we let the \scan operation return its view, i.e., the set obtained at line \ref{line:scan_lattice_renewals}. The view is a set of update commands from clients. Each element in the vector $V_p$ is a command. For an update command $up$ from a client, node $p$ invokes \update$(up)$. When receiving a query command from a client, node $p$ invokes the modified \scan() to return a set of commands and then apply these commands and return responses accordingly. 

\begin{algorithm*}[!hptb]
\begin{algorithmic}[1]

\begin{multicols*}{2}
\item[{\bf Upon receiving update $up$:}]
\STATE \update$(up)$
\STATE Respond ok to client 
\item[]

\item[{\bf Upon receiving query:}]
\STATE $view :=$ \scan$()$
\STATE $reply :=$ Apply($view$)
\STATE Respond $reply$ to client 
\end{multicols*}
\end{algorithmic}
\caption{Linearizable UQ State Machines}
\label{algo:linearizable_UQ_state_machine}
\end{algorithm*}

The following theorem implies that each command takes $O(1)$ rounds if there is no crash fault. $O(1)$ message size overhead means that if we assume that the size of each command is $O(1)$, then each message in the implementation has size $O(1)$. 

\begin{theorem}
There exists an implementation of linearizable update/query state machines such that each command takes $O(\sqrt{k})$ message delays and $O(1)$ message size overhead, where $k$ is the actual number of crash failures in the system. 
\end{theorem}

\section{Proofs for the ELA algorithm} 
\label{app:ELA}

\subsection{Proof of Lemma \ref{lem:comparable_views_for_same_node}} \label{app:comparable_views_for_same_node} 

\compViewsForSameNode*

\begin{proof}
The value of set $V_i[s]$ is modified only when $i$ receives a message from node $s$. Since $s$ is a non-faulty node and the communication is FIFO, the set $V_i[s]$ at time $t$ must be the same as the set $V_s[s]$ at some time $t_i$ and the set $V_j[s]$ at time $t'$ must be the same as the set $V_s[s]$ at some time $t_j$. The set $V_s[s]$ is non-decreasing. Thus, $V_i[s]$ at time $t$ must be comparable with $V_j[s]$ at time $t'$. 
\end{proof}

\subsection{Proof of Lemma \ref{lem:comparability}} \label{app:comparability} 

\comparabilityELA*

\begin{proof}
Let $V_i$ denote the vector at node $i$ and $V_j$ denote the vector at node $j$ when nodes $i$ and $j$ decide. The statement of the lemma is proved if we show that $V_i[i] $ and $V_j[j]$ are comparable. 

The decision condition on line \ref{ELA-line: wait until n-f} states that there exists a set $Q_i$ of size at least $n - f$ such that $V_i[i] = V_i[s]$ for each $s \in Q_i$ and a set $Q_j$ of size at least $n - f$ such that $V_j[j] = V_j[s]$ for each $s \in Q_j$. Since $f < \frac{n}{2}$, there exists a correct process $s \in Q_i \cap Q_j$. Lemma~\ref{lem:comparable_views_for_same_node} implies that $V_i[s]$ and $V_j[s]$ are comparable. This leads to the conclusion that $V_i[i]$ (which is equal to $V_i[s]$) is comparable to $V_j[j]$ (which is equal to $V_j[s]$)
%Applying Lemma~\ref{lem:comparable_views_for_same_node} we can see that at some time $t$, $V_i[i] = V_i[s]$ and at some time $t'$ $V_j[j] = V_j[s]$
\end{proof}

\subsection{Proof of Lemma \ref{lem:termination_without_newValue}} \label{app:termination_without_newValue} 
\terminationELA*

\begin{proof}
Let $C$ denote the set of correct nodes in an execution. Let node $i$ be an undecided node at time $t$ that does not crash by $t+2D$.
We show that by time $t+2D$, we have $V_i[j] = V_i[i]$ for each $j \in C$.
As a result, the predicate \EQ{} at line number~\ref{ELA-line: wait until n-f} becomes true, and node $i$ decides at Line~\ref{ELA-line: decide LA}.

Proof by contradiction. Suppose there exists a value $v \in V_i[i] - V_i[j]$ at time $t+2D$. Since by assumption, value $v$ is not an exposed value in this interval, it must be received by some correct node $s$ (or it is the input value of node $s$) at some time $t_s < t$.
Thus, value $v$ must be sent to all by node $s$ at time $t_s$ and received by all correct nodes  %by the end of interval $[t, t + D]$, 
by time $t_s+D < t+D$, including node $j$. By the algorithm, node $j$ must send value $v$ to all the other nodes before time $t+D$. 
Thus, node $i$ must receive value $v$ from node $j$ by  %end of interval $[t + D, t + 2D]$ 
 time $t+2D$, and add $v$ into $V_i[i]$ and $V_i[j]$, a contradiction to the assumption that $v \in V_i[i] - V_i[j]$.   
\end{proof}

\subsection{Proof of Lemma \ref{lem:failureChainLength}} \label{app:failureChainLength} 

\failureChainLength*

\begin{proof}
 Recall the definition of an exposed value $v$ occurring in an interval %, for an exposed value $v$ to occur in interval
 $(t,t+D]$: there has to be a failure chain ending with a correct process that receives value $v$ in interval $(t, t + D]$. Let $p_1,\cdots,p_{m-1}, p_m$ denote such a failure chain for value $v$, where $p_1\cdots p_{m-1}$ are faulty and node $p_m$ is non-faulty. 
%Since $t>D$, any exposed value $v$ introduced into the system in interval $[t,t+D]$ must be the input value of a faulty node. %By contradiction, if a failure chain does not exists, 

%Since $v$ is new in interval $[t, t + D]$, there exists a correct node $i$ which receives $v$ from some faulty node $j$ in interval $[t, t + D]$ and all other nodes in the chain must have crashed on line~\ref{line:send to all} of Algorithm~\ref{algo: lattice1}. %Let the length of this chain be $l$. 

Assume by contradiction that the length of this failure chain is $ l \le \frac{t}{D}$. Since $D$ is the message delay in the execution, node $p_i$ (in the failure chain) hears about $v$ at time $t_i \le i\cdot D$. Thus if $l \le \frac{t}{D}$, the correct node $p_m$ hears about $v$ at time $t_m \le \frac{t}{D} \cdot D = t$ making $v$ an exposed value occurring in an interval prior to $(t,t+D)$. This contradicts the assumption in the statement of the lemma that $v$ is an exposed value in interval $(t,t+D]$.
\end{proof}

\subsection{Proof of Lemma \ref{lem:unique_failure_chain}} \label{app:unique_failure_chain} 

\uniqueFailureChain*

\begin{proof}
Let $V$ and $U$ denote the set of the first $|P_v| - 2$ nodes in $P_v$ and the first $|P_u| - 2$ nodes in $P_u$. Suppose node $i \in V \cap U$ for contradiction. By condition (iv) of Definition \ref{def:failure_chain}, node $i$ crashes while sending $v$ to other nodes on line~\ref{ELA-line:echo to all} of ELA. Since lines \ref{line:ela_add} to~\ref{ELA-line:echo to all} of ELA are executed atomically, node $i$ cannot crash while sending value $u$ to other nodes at line~\ref{ELA-line:echo to all}, a contradiction. 
\end{proof}
\subsection{Proof of Lemma \ref{lem:round_complexity_single_shot}} \label{app:round_complexity_single_shot} 

\roundComplexityELA*

\begin{proof}
Let us assume that the algorithm takes $ 2 \sqrt{k} + 1 $ rounds for contradiction. By Lemma \ref{lem:termination_without_newValue},  we know that to prevent the algorithm from terminating, there has to be at least one exposed value every two rounds. Lemma~\ref{lem:failureChainLength} gives us the length of any failure chain and Lemma~\ref{lem:unique_failure_chain} states that a faulty node (except for the last 2 nodes in a failure chain) can be a part of only one failure chain. %Thus if an algorithm terminates in round $r$, we must have $r/2$ chains each of length $i = 1, 3, \cdots, r$. Thus the number of faulty nodes $= 1+3+\cdots r = \lfloor({r/2)}\rfloor^2$. }% $2 \sqrt{k} + 1$ the number of faulty nodes}% per failure chain if an a} 
%By Lemma \ref{lem:failureChainLength} and \ref{lem:unique_failure_chain}, 
Thus if the algorithm terminates in round $2 \sqrt{k} + 1 $,  the number of faulty nodes
must be at least $1 + 3 + \dots + 2\sqrt{k} - 1 > k$, leading to a contradiction. 
\end{proof}

\section{General Transformation in Message Passing Systems} \label{app:transformation_message}

In this section, we show how to adapt the transformation given by Attiya et al. in \cite{attiya1995atomic} for shared memory systems to work in message passing systems. The transformation algorithm, TS-ASO, is shown in Algorithm \ref{algo:transformation_message}. 

In the algorithm, each node $p$ keeps track of a vector $Snap$ of size $n$, which is the local view  of the shared object, i.e., $Snap[q]$ stores the most recent value written by node $q$ known by node $p$. The variable $V$ is a map from tag number to snapshot. $V[r]$ is the snapshot vector obtained for tag $r$. Variable $r$ denotes the tag number of the last lattice agreement instance that node $p$ has completed. $maxTag$ keeps track of the largest tag ever seen by a node. $ts$ is a sequence number for values, which is increased by one when a new value needs to be written. Each value is associated with a timestamp. For any value $v$, we use $ts_v$ to denote its associated timestamp. If a variable belongs to node $i$, we use the subscript $i$ to denote it. For example, $maxTag_i$ denotes the value of variable $maxTag$ at node $i$. 

{\sc \bf Scan()} operation: a \scan operation invokes at most two lattice operations to obtain a view. We call the two {\bf for} loops as its two phases. At each phase, it first decides which lattice agreement instance to run by reading the largest tag from at least $n - f$ nodes at line \ref{line:read_tag}. Then, it writes the tag obtained to at least $n - f$ nodes at line \ref{line:write_tag}. At line \ref{line:read_state}, it reads the local state from at least $n - f$ nodes and the join of all states received will be used as input for its lattice agreement. At line \ref{line:write_state}, it writes the join of all states read at line \ref{line:read_state} to at least $n - f$ nodes. Then, it invokes the lattice agreement instance with the tag obtained at line \ref{line:read_tag} and the vector obtained at line \ref{line:read_state} as input parameters. After completion of the lattice agreement, if it does not observe a higher tag number, then it directly returns the view obtained from its lattice agreement invocation. If it does not return a view after lattice agreement invocation, it borrows a view from some other node, which is guaranteed to exist. 

{\sc \bf Update$(v)$} operation: To write value $v$, node $p$ increases $ts_p$ by one and assign it to be the timestamp of $v$. It writes value $v$ with its timestamp to at least $n - f$ nodes. Then, it executes a \scan operation. The view returned by the \scan operation is not used. The \scan operation acts as a synchronization point.

\begin{algorithm*}[!hptb]
\begin{algorithmic}[1]
	\small
	\item[{\bf Local Variables:}] 
	
	\item[] $Snap$ \COMMENT{vector of size $n$, local view of the shared object.} 
	
	\item[] $V$ \COMMENT{snapshots obtained. $V[r]$ is the snapshot obtained for tag $r$}
    \item[] $r$ \COMMENT{tag number for lattice agreement}
    
    \item[] $maxTag$ \COMMENT{Integer, largest tag ever seen} 
    
    \item[] $ts$ \COMMENT{timestamp for value}
    \item[] \hrulefill

\begin{multicols*}{1}
\item[\textbf{Procedure Scan()}:]
\FOR{$phase :=$ 1 to 2} 
\STATE $readTag()$ \label{line:read_tag}
\STATE $r \gets \max (maxTag, r + 1)$  \label{line:tag_for_lattice}
\STATE $writeTag(r)$ \label{line:write_tag} 
\STATE $input \gets readState()$ \label{line:read_state}
\STATE $writeState(input)$ \label{line:write_state}
\STATE $output := LA(r, input)$ \label{line:lattice_op}
\STATE $readTag()$ \label{line:read_tag_again}
\IF{$maxTag \leq r$} \label{line:no_higher_tag}
\STATE $V[r] := output$ 
\STATE $writeView(output, r)$ \label{line:write_view}
\STATE {\bf return} $V[r]$ 
\ELSIF{$phase = 2$} 
\STATE {\bf wait until} $V[r] \not = \emptyset$ \label{line:borrow_view}
\STATE {\bf return} $V[r]$
\ENDIF
\ENDFOR

\item[]
\item[]

\item[\textbf{Procedure Update($v$)}:]
\STATE $ts \gets ts + 1$ \\
\STATE $writeValue(\tuple{v, ts})$
\STATE {\bf return {\em Scan}()} \\
\item[] 
\item[]

\item[\textbf{Procedure readTag()}:]
\STATE Send $(\quotes{readTag})$ to all 
\STATE \textbf{Wait until} receiving \\\hspace{0.3in} $\geq n-f$ $(\quotes{readTagAck}, *)$ msgs
\STATE $maxTag \gets$ largest tag received%\\
%\item[\textbf{End: readTag()}]

\item[] 

\item[\textbf{Procedure writeTag($tag$)}:]
\STATE Send $(\quotes{writeTag}, tag)$ to all 
\STATE \textbf{Wait until} receiving \\\hspace{0.3in} $\geq n-f$ $(\quotes{writeTagAck}, tag)$ msgs
%\item[\textbf{End: writeTag}]

\item[] 

\item[\textbf{Procedure writeValue($\tuple{v, ts}$)}:]
\STATE Send ($\quotes{value}$, $\tuple{v, ts}$) to all
\STATE \textbf{Wait until} receiving \\\hspace{0.3in} $\geq n - f$ ($\quotes{valueAck}$) msgs \\

\item[] 

\item[\textbf{Procedure readState()}:]
\STATE Send ($\quotes{readState}$) to all
\STATE \textbf{Wait until} receiving \\\hspace{0.3in} $\geq n - f$ ($\quotes{readStateAck}$, $*$) msgs \\ 
\STATE Let $S_j$ denote the state vector received from $j$
\STATE {\bf return} $\bigsqcup \limits_{j = 1}^{n} S_j$

\item[] 

\item[\textbf{Procedure writeState($S$)}:]
\STATE Send ($\quotes{writeState}$, $S$) to all
\STATE \textbf{Wait until} receiving \\\hspace{0.3in} $\geq n - f$ ($\quotes{writeStateAck}$) msgs \\

\item[]

\item[\textbf{Procedure writeView($view$, $r$)}:]
\STATE Send ($\quotes{writeView}$, $view$, $r$) to all
\STATE \textbf{Wait until} receiving \\\hspace{0.3in} $\geq n - f$ ($\quotes{viewAck}$) msgs \\

\item[]

\item[/* Event handlers: executing in background */]
\item[/* All event handlers executed atomically */]

\item[{\bf Upon receiving $(\quotes{value}, \tuple{u, ts})$ from $q$}:]
\STATE  $Snap[q] \gets \max(Snap[q], \tuple{u, ts})$
\item[]

\item[{\bf Upon receiving $(``writeTag", tag)$} from $q$:]
\STATE $maxTag \gets \max(maxTag, tag)$
\STATE Send $(\quotes{writeTagAck}, tag)$ to $q$ 
\item[]

\item[{\bf Upon receiving $(\quotes{readTag})$ from $q$}:]
\STATE Send $(``readTagAck", maxTag)$ to $q$
\item[]

\item[{\bf Upon receiving $(\quotes{readState})$ from $q$}:]
\STATE Send $(\quotes{readStateAck}, Snap)$ to $q$

\item[]

\item[{\bf Upon receiving $(\quotes{writeState}, S)$ from $q$}:]
\STATE $Snap \gets Sanp \sqcup S$
\STATE Send $(\quotes{writeStateAck})$ to $q$

\item[]

\item[{\bf Upon receiving $(\quotes{writeView}, U, r)$ from $q$}:]
\STATE $V[r] \gets V[r] \sqcup U$

\item[]

\end{multicols*}
\end{algorithmic}
\caption{TS-ASO: code for node $p$.}
\label{algo:transformation_message}
\end{algorithm*}

\subsection{Proof of Correctness} 
In this section, we show that atomic object implementation is linearizable by explicitly constructing a linearization of all update and scan operations. We can first obtain the following lemma regarding to the views returned by all operations. We say a view returned by an operation is a $direct$ view if this view is returned in the first phase of the scan procedure. Otherwise, we call this view as an $indirect$ view. We can readily see that a direct view of $p_i$ is obtained from an execution of lattice agreement of $p_i$. An indirect view of node $i$ is direct view of some other node $j$. We call an invocation of a lattice agreement with tag $r$ as a lattice operation with tag $r$.

\begin{lemma} \label{lem:input_quorum}
Consider a lattice operation $op$ by node $p$ with tag $r$, suppose $op$ returns $V_{op}$ at time $t$. Then, for each $j$, there exists a set of nodes $Q_j$ with size at least $n - f$ such that $V_{op}[j] \leq Snap_i[j]$ for each $i \in Q_j$ at time $t$. 
\begin{proof}
Let $P$ denote the set of nodes which invokes the lattice operation $op$ before or at time $t$. By Upward-Validity, we have that $V_{op}[j] = input_p[j]$ for some node $p \in P$. By line \ref{line:write_state}, there exists a set of nodes $Q_j$ with size at least $n - f$ such that $input_p \leq Snap_q$ for each $i \in Q_j$. Therefore, $V_{op}[j] = input_p[j] \leq Snap_i[j]$ for any $i \in Q_j$.  
\end{proof}
\end{lemma}

\begin{lemma}
If two operations return $view_i$ and $view_j$, then $view_i$ and $view_j$ are comparable. 
\begin{proof}
We only need to show that $view_i$ and $view_j$ are comparable if they are direct views. Let $op_i$ and $op_j$ denote the two operations that return $view_i$ and $view_j$, respectively. We have the following cases. 

Case 1. $view_i$ and $view_j$ are obtained from the same lattice operation. By comparability of lattice agreement, $view_i$ and $view_j$ are comparable. 

Case 2. $view_i$ are obtained from lattice operation with tag $r_i$ and $view_j$ is obtained from lattice operation with round $r_j$. Assume that $r_j > r_i$, w.l.o.g. Assume that $op_i$ obtains $view_i$ in {\bf first phase}. The case that $op_i$ returns $view_i$ in {\bf second phase} is symmetric. Then, in line $op_i$ finds no nodes with tag number greater than $r_i$. Therefore, $op_i$ obtains $view_i$ before $op_j$ completes line \ref{line:write_tag}. Then, when $op_j$ starts to read states of at least $n - f$ nodes at line \ref{line:read_state}, it must be able to read $view_i[j]$ for each $j$, by Lemma \ref{lem:input_quorum} and quorum intersection. Thus, $view_i$ is less that or equal to the input of $op_j$ for the lattice operation with tag $r_j$. By Downward-Validity of lattice agreement, we have $view_i \leq view_j$. 
\end{proof} 
\end{lemma}

Now, we associate an view with the beginning of an operation. For operation $op$, the view associated with the beginning of $op$ is view $view'$ such that $view'[j]$ is the largest value written by $p_j$ which is contained in the local state of at least $n - f$ nodes. 

\begin{lemma} \label{lem:view_begin} 
Assume operation $op$ returns $view$ and let $view'$ denote the view associated with the beginning of $op$, then $view' \leq view$.
\begin{proof}
Consider the following two cases. 

Case 1. $op$ return $view$ directly from lattice operation with tag $r$. W.l.o.g, assume that $op$ returns $view$ in the {\bf first phase}. The case where $op$ returns $view$ in the {\bf second phase} is symmetric. Let $input$ denote the input of $op$ for the lattice operation with tag $r$ at line \ref{line:lattice_op}. By definition of $view'$, when $op$ executes line \ref{line:read_state}, it must be able to read all values in $view'$. Since nodes write increasing values to the snapshot object, $view' \leq input$. By Downward-Validity of lattice agreement, we have $input \leq view$. Thus, $view' \leq view$. 

Case 2. $op$ returns $view$ indirectly. Then, $op$ must continue to phase 2. Let $r_1$ and $r_2$ denote the tag number $op$ obtains at line \ref{line:tag_for_lattice} of phase 1 and phase 2, respectively. We have $r_1 < r_2$. Consider the second phase, the condition at line \ref{line:no_higher_tag} is satisfied and $op$ borrows a $view$ of some other nodes for tag $r_2$. $view$ must be a direct view of some operation $op'$ for tag $r_2$. W.l.o.g, assume that $op'$ returns $view$ in the {\bf first phase}. The case where $op'$ returns $view$ in the {\bf second phase} is symmetric. Since $r_1 < r_2$, $op'$ must start line \ref{line:read_state} after $op$ starts. Otherwise, $op$ would obtain tag number $r_2$ instead of $r_1$ for its first phase. By the definition of the view associated with an operation, $op'$ must be able to read all values in $view'$ at line \ref{line:read_state}. Downward-Validity of lattice agreement implies that $view' \leq view$. 
\end{proof}
\end{lemma}

\begin{lemma}
Consider two operations $op_i$ and $op_j$ that return $view_i$ and $view_j$, respectively. If $op_i \rightarrow op_j$, then $view_i \leq view_j$. 

\begin{proof}
Let $view$ be the view associated with the beginning of $op_j$. By Lemma \ref{lem:view_begin}, $view \leq view_j$. Since $op_i$ obtains $view_i$ before the beginning of $op_j$, Lemma \ref{lem:input_quorum} implies that $view_i \leq view$. Thus, $view_i \leq view_j$. 
\end{proof}
\end{lemma}

\begin{lemma}
Let $up$ be an \update{} operation by node $p$ that writes value $v$, and returns $view_p$. Then, $\tuple{v, ts_v} \leq view_p[p]$.
\begin{proof}
Let $view$ be the view associated with the beginning of the \scan operation embedded in $up$. Since $up$ writes $\tuple{v, ts_v}$ to at least $n - f$ nodes before its embedded \scan operation, then $\tuple{v, ts} \leq view[p]$. By Lemma \ref{lem:view_begin}, we have $view \leq view_p$. Thus, $\tuple{v, ts_v} \leq view_p[p]$.  
\end{proof}
\end{lemma}

The linearization sequence of the $scan$ and $update$ operations is constructed in the same way as the one given in \cite{attiya1995atomic}. First, we construct a sequence $\sigma'$ which only includes all $scan$ operations. The sequence also includes the $scan$ operations embedded in the $update$ operations. The $scan$ operations are ordered in $\sigma'$ according to the order of the views returned by them. Specifically, for any two $scan$ operations $sc_i$ and $sc_j$ that return $view_i$ and $view_j$, respectively, if $view_i < view_j$, then $sc_i$ appears before $sc_j$ in $\sigma'$. If $view_i = view_j$ and $sc_i \rightarrow sc_j$, then $sc_i$ appears before $sc_j$ in $\sigma'$. Otherwise, $sc_i$ and $sc_j$ are ordered arbitrarily. 

Now we create a linearization sequence $\sigma$ from $\sigma'$ by inserting all $update$ operations into $\sigma'$. Consider an operation $op_i$ that writes value $v$. We insert $op$ after all $scan$ operations that return a value strictly smaller than $v$ and before all scan operations that return a value greater than or equal to $v$. That is, $op$ is inserted  just before the first $scan$ operation that returns a view which contains $v$. For any two operations $op_1$ and $op_2$ that fit between the same pair of $scan$ operations. If $op_1 \rightarrow op_2$, then we put $op_1$ before $op_2$ in sequence $\sigma$. Otherwise, $op_1$ and $op_2$ are ordered arbitrarily. 

As long as we have the above lemmas, the proof which shows that $\sigma$ is a linearization is the same as the proof in \cite{attiya1995atomic}. 

\begin{theorem}
There exists an atomic snapshot object implementation in asynchronous crash-prone message passing systems, which requires $O(\log n)$ message delays per $update$ or $scan$ operation, where $f < \frac{n}{2}$ is the maximum number of crash failures in the system. \begin{proof}
The paper~\cite{zheng2018linearizable} presents an $O(\log f)$ rounds algorithm for the lattice agreement problem in asynchronous crash-prone message passing systems. Directly plugging in their algorithm into our transformation result in $O(n)$ rounds complexity for update and scan operations, since their algorithm requires all nodes to start around the same time. Their algorithm can be simply modified to run in $O(\log n)$ rounds even if nodes start at different times. 
\end{proof}
\end{theorem}

\section{Comparison between TS-ASO and AC-ASO} \label{app:comparison_TS-ASO_AC-ASO}
Let TS-ASO denote the general transformation in Appendix \ref{app:transformation_message}. We list the two primary differences between TS-ASO and AC-ASO here. 

1) In both TS-ASO and AC-ASO, to write a new value (in an \update{} operation), a node needs to first read the largest tag from a quorum of nodes. Let $r$ denote the tag obtained. In TS-ASO, a node directly participates in the lattice operation with tag $r$ (This design ensures that there is good lattice operation for each tag). This design is also the main reason why TS-ASO cannot preserve the round complexity of our ELA algorithm, since the round complexity of our ELA algorithm depends on the assumption that each node starts the algorithm around the same time and nodes can join the lattice operation with the same tag at quite different times. To solve this problem, our idea is to let a node participate a lattice operation with a strictly greater tag than the tag it reads from a quorum. That is, if a node observes a tag $r$, then it participates the lattice operation with tag $r + 1$. This ensures that all nodes participate the lattice operation with same tag around the same time (at most constant round apart). If we only have the above modification, we cannot guarantee that there exists a good lattice operation for each tag, then when some lattice operation needs to borrow a view from a good lattice operation, the existence of such a good lattice operation is not guaranteed. Thus, to tackle this problem, our idea is to use a dummy lattice operation whose only purpose is to ensure the existence of a good lattice operation for each tag. Specifically, we let each \update{} operation executes an initial lattice operation with tag $r$ but without introducing a new value with tag $r$. This initial lattice operation guarantees the existence of a good lattice operation for each tag but does not prevent the termination of existing lattice operations with the same tag due to the reason below. Our lattice operation has the following properties: the termination of a lattice operation with tag $T$ depends on only the values with tags at most $T$. Since the initial lattice operation with tag $r$ does not introduce a new value with tag $r$, it does not prevent progress of other existing lattice operations with tag at most $r$. This is also the reason why our design is not a general transformation that preserves the round complexity of any lattice agreement algorithm. 

2) In TS-ASO, before participating a lattice operation, a node needs to collect the states of at least a quorum of nodes and use their join as input for the lattice operation. In TS-ASO, such a read step is important in ensuring the correctness. (First, it ensures that a lattice operation with a bigger tag must be able to read the view obtained by a lattice operation with a smaller tag. Second, it ensures that a later \update{} or \scan operation must be able to read the view obtained by previous (completed) \update{} or \scan operation.) Each state read from other nodes is a vector of $n$ values. Thus, each message in the lattice operation has $O(n)$ overhead in size (it contains at least $n$ values). In our design, we would like to remove such overhead in message size. Note that our ELA algorithm has constant message size overhead. In AC-ASO, when a node participates a lattice operation, it does not collect the states of at least a quorum of nodes and use that as input for lattice operation. Without the reading step, two lattice operations are not sufficient to guarantee correctness. To ensure correctness, we will show that three lattice operations are sufficient. 

\section{Proofs for the \alg{} Algorithm} \label{app:aso_proofs}
 
\subsection{Mesage Handlers of \alg{}} \label{app:message_handlers}

\begin{itemize}
\vspace*{-0.1in}
    \item Upon receiving a $\quotes{writeTag}$ message: node $i$ updates its $maxTag$ to be the maximum of its current $maxTag$ and the received tag. Then, it responds a $writeAck$. 
    
    % In Lemma XXX, we rely on this design to XXX. \lewis{Explain why we need to do this later.} \xiong{put this to somewhere else since \quotes{echoTag} message can also update tag} \lewis{Yes, I haven't figured out a good place to have this. }. \xiong{I put it in the $\quotes{value}$ message handler and adjusted the order of handlers. I think maybe we don't need to explain why we do not update maxTag when receiving value message}
    
    \item Upon receiving a $\quotes{echoTag}$ message: node $i$ updates its $maxTag$ to be the maximum of its current $maxTag$ and the received tag. 
    
    \item Upon receiving a $\quotes{value}$ message from node $j$: node $i$ adds the value into $V_i[i]$ and $V_i[j]$. It then forwards this value to all other nodes if it has never done so before. It is important to note that a node does \textit{not} update its $maxTag$ variable when it receives a value with a larger tag from a \quotes{value} message. The $maxtTag$ variable is only updated when a node receives a $\quotes{writeTag}$ message or $\quotes{echoTag}$ message. 
    
    \item Upon receiving a $\quotes{readTag}$ message: node $i$ responds a $readAck$ message along with the largest tag it has ever seen via $writeTag$ messages. 
    
    \item Upon receiving a $\quotes{goodLA}$ message with tag $r$ from node $j$: node $i$ \textit{borrows} the view from node $j$ by recording $V_i[j]^{\leq r}$. Our design ensures that the borrowed view is identical to the view from $j$'s \textit{good} lattice operation.
    By assumption the communication is FIFO, and node $j$ sends the message $(\quotes{goodLA}, r)$ right after its view satisfies the equivalence predicate at line \ref{line:lattice_termination_condition}; thus, $V_i[j]^{\leq r}$ must be the same as the view of the particular good lattice operation at node $j$. We formally prove this claim in Lemma \ref{lem:view_from_good_lattice}.
\end{itemize}

\subsection{Proof of Lemma \ref{lem:operation_termination}} \label{app:operation_termination}

\terminationASOLemma*

\begin{proof}
We show that each $LatticeRenewal()$ invocation terminates. 
Let $Op$ denote the $LatticeRenewal()$ procedure at node $i$. 
We only need to show that the condition at line \ref{line:wait_borrow} is eventually satisfied if $Op$ has not returned earlier, since the condition is the only blocking code inside $Op$. Consider the phse 3 lattice operation $L_3$. Since $Op$ continues to line \ref{line:wait_borrow} with $phase = 3$, $L_3$ is \textit{not} a good lattice operation. This means that the {\bf for} loop of $Op$ breaks at line \ref{line:bad_lattice_op}; hence, $r_i$ at line \ref{line:wait_borrow} is equal to the tag used by $L_3$. Since $L_3$ is not a good lattice operation, and it returns $false$, we have $maxTag_i > r_i$ at line \ref{line:good_lattice_condition}. In other words, at this point of time, the largest tag in the system is at least $maxTag_i$. Lemma \ref{lem:nonskipping_good_execution} implies that a good lattice operation with tag $r_i$ must be completed before $L_3$ executes line \ref{line:good_lattice_condition}. By assumption of the reliable communication channel, node $i$ is able to receive a $(\quotes{completed}, r_i)$ message from the good lattice operation.
After receiving the message, condition at line \ref{line:wait_borrow} is satisfied, and hence, $LatticeRenewal()$ terminates.
\end{proof}

\subsection{Proof of Lemma \ref{lem:view_from_good_lattice}} \label{app:view_from_good_lattice}

\viewFromGoodLattice*

\begin{proof}
If the operation is direct, then by definition, its view is the view of its final lattice operation in $LatticeRenewal()$ procedure. Otherwise, the operation borrows the view from some other node $j$ at line \ref{line:wait_borrow}, which must be the view of $j$'s good lattice operation. This is because only a good lattice operation sends a $\quotes{goodLA}$ message. 
\end{proof}

\subsection{Proof of Lemma \ref{lem:comparable_tagged_views_for_same_node}} \label{app:comparable_tagged_views_for_same_node} 

\comparableTagViewsSameNode*

\begin{proof}

The value of set $V_i[s]^{\leq T}$ is modified only when $i$ receives a new value with tag $\leq T$ from node $s$. Since the communication is FIFO, the set $V_i[s]^{\leq T}$ at time $t$ must be the same as the set $V_s[s]^{\leq T}$ at some time $t_i$ and the set $V_j[s]^{\leq T}$ at time $t'$ must be the same as the set $V_s[s]^{\leq T}$ at some time $t_j$. The set $V_s[s]^{\leq T}$ is non-decreasing. Thus, $V_i[s]^{\leq T}$ at time $t$ must be comparable with $V_j[s]^{\leq T}$ at time $t'$. 

% For the sake of contradiction, suppose there exists a value $v \in V_i[s]^{\leq T} - V_j[s]^{\leq T}$ and a value $u \in V_j[s]^{\leq T} - V_i[s]^{\leq T}$. For node $s$, it either sends value $v$ before $u$ or $u$ before $v$. Suppose node $s$ sends $v$ before $u$, w.l.o.g. Because the communication channel is FIFO, node $j$ must receive value $v$ before receiving $u$. This contradicts the assumptions that (i) there exists a value $v \in V_i[s]^{\leq T} - V_j[s]^{\leq T}$; and (ii) $u \in V_j[s]^{\leq T}$. Thus, $V_i[s]^{\leq T}$ and $V_j[s]^{\leq T}$ must be comparable. 
\end{proof}

% \subsection{Proof of Lemma \ref{lem: sequential specification proof}} \label{app:sequential specification proof} 

% \subsection{Proof of Lemma \ref{lem: real time ordering}} \label{app:read time ordering} 

\subsection{Proof of Lemma \ref{lem:lattice_comparable_view}} \label{app:lattice_comparable_view}

\laComparableViews*

\begin{proof}
Consider two good lattice operations $Op_i$ with tag $T_i$ by node $i$  and $Op_j$ with tag $T_j$ by node $j$. Let $H_i$ and $H_j$ denote node $i$ and $j$'s view, respectively, after they complete line \ref{line:lattice_termination_condition}. Recall that by Definition \ref{def:view}, $H_i = V_i^* = V_i[i]^{\leq T_i}$ and $H_j = V_j^* = V_j[j]^{\leq T_j}$ right after the equivalence quorum predicate is satisfied.

To prove the lemma, we need to show that either $H_i \subseteq H_j$ or $H_j \subseteq H_i$. Assume without loss of generality $T_i \leq T_j$. Then consider two following cases. 
\begin{itemize}
    \item \textit{Case 1}: $T_i = T_j = T$. 
    Intuitively, both nodes participate in the same instance of lattice agreement, and thus, they will obtain  comparable outputs (views).
    
    Formally, 
    let $W_i$ and $W_j$ denote the equivalence quorum of lattice operation $Op_i$ and $Op_j$, respectively. Thus, there exists a nonfaulty node $s \in W_i \cap W_j$ such that $H_i = V_i[i]^{\leq T} = V_i[s]^{\leq T}$ and $H_j = V_j[j]^{\leq T} = V_j[s]^{\leq T}$. Lemma \ref{lem:comparable_tagged_views_for_same_node} implies that $V_i[s]^{\leq T}$ and $V_j[s]^{\leq T}$ are comparable. Thus, $H_i$ must be comparable with $H_j$, since to satisfy the equivalence quorum predicate, $H_i = V_i[i]^{\leq T} = V_i[s]^{\leq T}$ and $H_j = V_j[j]^{\leq T} = V_j[s]^{\leq T}$.
    
    \item \textit{Case 2}: $T_i < T_j$. In this case, we show that $H_i \subseteq H_j$.
    Roughly speaking, we want to show that lattice operation with a larger tag start with a view that is  at least as large as the view of any lattice operation with a smaller tag.
    we rely on Property 2 of $LatticeRenewal()$ and the way we update $maxTag$'s to prove the claim. 
    
    We make the following observations:
    
    \begin{itemize}
        \item \textit{Obs. 1}: the tag $T_j$ is \textit{not} known by node $i$ when $Op_i$ completes line \ref{line:lattice_termination_condition}.
        
        This is because (i) $Op_i$ is a good lattice operation, and the condition $maxTag_i \leq r_i$ at line \ref{line:good_lattice_condition} of $Op_i$ must return true; and (ii) line \ref{line:lattice_termination_condition} and \ref{line:good_lattice_condition} are executed atomically, and hence $V_i^* = V_i[i]^{\leq T_i}$ does \textit{not} change during this block of code. \vspace{5pt}
        
        \item \textit{Obs. 2}: There exists a node $s$ such that $H_i \subseteq V_s[s]^{\leq T_i}$ and $V_s[s]^{\leq T_i} \subseteq V_j[j]^{\leq T_j}$.
        
        Let $W_i$ denote the equivalence quorum of $Op_i$. Let $Q_j$ denote the set of at least $n - f$ nodes that sent the $\quotes{writeAck}$ messages in responding to the $\quotes{writeTag}$ message of $Op_j$. Since both set are of size at least $n-f$, there exists a node $s \in W_i \cap Q_j$ such that (i) $H_i = V_i^* =  V_i[s]^{\leq T_i}$ (due to the equivalence predicate); and (ii) node $s \in Q_j$. 
        
        By \textit{Obs. 1}, and the assumption of FIFO communication, node $s$ must have received all values in $V_i[s]^{\leq T_i}$ before receiving the $\quotes{writeTag}$ message of $Op_j$. Otherwise, $Op_i$ would observe tag $T_j$ at line \ref{line:good_lattice_condition}.  
        % Since the tag $T_j$ is not known by node $i$ when $Op_i$ completes line \ref{line:lattice_termination_condition}, and the communication is assumed to be FIFO, node $s$ must receive all values in $V_i[s]^{\leq T_i}$ before receiving the $\quotes{writeTag}$ message of $Op_j$. 
        Thus, node $s$ must send out all values in $V_i[s]^{\leq T_i}$ to all before sending $writeAck$ for the $\quotes{writeTag}$ message of $op_j$. Then, when node $j$ receives the $writeAck$ from node $s$, it must also receive all values in $V_i[s]^{\leq T_i}$, i.e., the view of node $j$ contains all values in $V_i[s]^{\leq T_i}$ after line \ref{line:lattice_writeTag} of $op_j$ completes. Since $H_j$ is the set of values with tag at most $T_j > T_i$ in the history of node $j$ when $op_j$ completes line \ref{line:lattice_termination_condition},  $V_i[s]^{\leq T_i} \subseteq H_j$. Therefore, $H_i \subseteq H_j$.
        
        % \lewis{Xiong: I don't quite get this sentence. Since $s$ is sending $\quotes{writeAck}$ messages, it could send this message during the time it is sending values in $V_i[s]^{\leq T_i}$ to node $i$. What am I missing? I changed the wording a bit. You can check aso\_old.tex for old text. Maybe we need to add a message handler to increment $maxTag$ when receiving $\quotes{writeAck}$ messages?}
    \end{itemize}

    % \lewis{+++++ Need to check whether the new edit fixes the problem or not +++++}

\end{itemize}
\end{proof}

\subsection{Proof of Theorem \ref{theo:mainTheoremASO}} \label{app:mainTheoremASO} 

We first show the following two lemmas. 

\begin{restatable}{lemma}{semanticsASO} \label{lem: sequential specification proof}

The sequence $\sigma$ preserves the semantics of atomic snapshot object, i.e., for any \scan operation which returns $Snap$, for any $i$, $Snap[i]$ must be the value written by the latest \update{} operation of node $i$ that appears before the \scan operation in $\sigma$. 
\end{restatable}

\begin{proof}
Assume $Snap[i] = v$, and let $op$ be the \update$(v)$ operation by node $i$. By the construction of $\sigma$, $op$ appears before $sc$ in $\sigma$. First, we have that any \update{} operation by node $i$ that writes a value strictly greater than $v$ is ordered after $sc$ in $\sigma$. Furthermore, any \update{} operation by node $i$ that writes a value strictly smaller than $v$ is ordered before $op$ in $\sigma$. Therefore, $op$ is the last \update{} operation by node $i$ that is ordered before \scan.
\end{proof}

The following lemma is implied by Lemma \ref{lem:output_view_dominate_input_view}. 
\begin{restatable}{lemma}{downwardValidity} \label{lem:downward_validity}
Let $Op$ be an \update{} operation by node $i$ that writes value $v$ with timestamp $ts$ and has view $H_i$. Then, $\tuple{v, ts} ~ \in H_i$.
\end{restatable}

\begin{restatable}{lemma}{realTimeOrder} \label{lem: real time ordering}
The sequence $\sigma$ respects the real-time order of operations, i.e., for any two operations $Op_i$ and $Op_j$, if $Op_i \rightarrow Op_j$, then $Op_i$ appears before $Op_j$ in $\sigma$.
\end{restatable}

\begin{proof}

Let $H_i$ and $H_j$ be the views of $op_i$ and $op_j$, respectively. Let node $i$ and node $j$ denote the nodes where $op_i$ and $op_j$ take place, respectively. Note that $i$ may be equal to $j$. By Lemma \ref{lem:view_order}, we have $H_i \subseteq H_j$. 
We consider the following cases. 
\begin{itemize}

    \item[$op_i=$ \update$(v)$ and $op_j=$ \update$(u)$:] 
   If $op_i$ and $op_j$ are placed between the same pair of \scan operations, then they are ordered according to $\rightarrow$. Hence, $op_i$ appears before $op_j$ in $\sigma$. Otherwise, there exists a \scan operation $sc$ with view $H_{sc}$ between $op_i$ and $op_j$ in $\sigma$. Suppose, by way of contradiction, that $op_j$ is ordered before $op_i$. In other words, $sc$ appears after $op_j$ and appears before $op_i$ in $\sigma$. Then $\tuple{v, ts_v} ~\not \in H_{sc}$ and $\tuple{u, ts_u} ~\in H_{sc}$. Lemma \ref{lem:downward_validity} implies that $\tuple{v, ts_v} ~ \in H_i$. On the other hand, since $op_i \rightarrow op_j$, we have $\tuple{u, ts_u} ~\not \in H_i$. Thus,  $H_{sc}$ and $H_i$ are incomparable, a contradiction to Lemma \ref{lem:comparable_views}.
   
    \item[$op_i=$ \update$(v)$ and $op_j=$ \scan:] 
    
    Lemma \ref{lem:downward_validity} and Lemma \ref{lem:view_order} together imply that $~\tuple{v, ts_v} \in H_i \subseteq H_j$. Thus, $op_i$ must be ordered before $op_j$ in $\sigma$. 
   
    \item[$op_i=$ \scan and $op_j=$ \scan:]
    
    If $H_i \not = H_j$, which means $H_i \subset H_j$. In this case, $op_i$ is ordered before $op_j$ by construction. Otherwise, since two \scan operations that return the same view are ordered according to $\rightarrow$, $op_i$ is ordered before $op_j$ in $\sigma$. 
    
    \item[$op_i=$ \scan and $op_j=$ \update$(v)$:]
    
    Clearly, $\tuple{v, ts_v} ~ \not \in H_i$. Since $op_j$ is ordered after all \scan operations whose view does not contain $\tuple{v, ts_v}$, it follows that $op_i$ appears before $op_j$ in $\sigma$. 
\end{itemize}
\end{proof}

\mainTheoremASO*

\begin{proof}
Immediately follows from Lemma \ref{lem: sequential specification proof} and \ref{lem: real time ordering}. 
\end{proof}

\subsection{Proof of Lemma \ref{lem:later_operation_bigger_tag}} \label{app:later_operation_bigger_tag} 

\laterOperationBiggerTag*

\begin{proof}
Let operation $op$ denote such a lattice operation. Since $op$ sends its tag $T$ to all in the $writeTag$ function at line 28, by time $t + D$, each correct node must have received tag $T$. Thus, any \update{} operation that starts after time $t + D$ must obtain a tag greater than $T$ for its value. Thus, all values with tags at most $T$ must come from \update{} operations that start before time $t + D$. Since local computation does not take time, all such values must be sent out at line \ref{line:brodcast_val} before time $t + D$. 
\end{proof}

\subsection{Proof of Lemma \ref{lem:completed_update_known_to_all_after_D}} \label{app:completed_update_known_to_all_after_D}

\knownToAllAfterD*

\begin{proof}
Since $op$ completes before time $t$, $\tuple{v, ts_v}$ must be sent to node $q$ before time $t$ and must be received by node $q$ by time $t + D$. Thus, node $p$ must receive $\tuple{v, ts_v}$ from node $q$ by the end of $t+ 2D$. 
\end{proof}

\end{document}